\documentclass[11pt,twoside]{article}
\usepackage{hyperref,xspace,fullpage}
\usepackage{vitalymacros}
\usepackage{amsmath}
\usepackage{amsfonts}
\usepackage{paralist}
\usepackage{bm}

\newif\iffull
\fulltrue

\newtheorem{fact}{Fact}[section]
\newtheorem{definition}[fact]{Definition}

\newtheorem{theorem}[fact]{Theorem}
\newtheorem{lemma}[fact]{Lemma}
\newtheorem{corollary}[fact]{Corollary}

\newcommand{\on}{\{-1,1\}}
\newcommand{\infl}{\mathsf{Inf}}
\newcommand{\aedeg}{\mathsf{deg}^{\ell_2}_\epsilon}

\newcommand{\cF}{{\cal F}}
\newcommand{\lcond}{\ \left|\ }

\def\b1{{\bf 1}}

\newcommand{\RR}{\mathbb R}

\newenvironment{proof}{\noindent \textbf{Proof:}}{\hfill{$\Box$}}

\title{Tight Bounds on Low-degree Spectral Concentration of Submodular and XOS Functions}

\author{Vitaly Feldman \\
IBM Research - Almaden
 \and Jan Vondr\'{a}k \\
IBM Research - Almaden\\
}

\begin{document}

\maketitle

\begin{abstract}
Submodular and fractionally subadditive (or equivalently XOS) functions play a fundamental role in combinatorial optimization, algorithmic game theory and machine learning. 
Motivated by learnability of these classes of functions from random examples, we consider the question of how well such functions can be approximated by low-degree polynomials in $\ell_2$ norm over the uniform distribution. This question is equivalent to understanding the concentration of Fourier weight on low-degree coefficients, a central concept in Fourier analysis.
Denoting the smallest degree sufficient to approximate $f$ in $\ell_2$ norm within $\eps$ by $\aedeg(f)$, we show that
\begin{itemize}
\item For any submodular function $f:\zo^n \rightarrow [0,1]$, $\aedeg(f)=O(\log (1/\epsilon)/\epsilon^{4/5})$ and there is a submodular function that requires degree $\Omega(1/\epsilon^{4/5})$.
\item For any XOS function $f:\zo^n \rightarrow [0,1]$, $\aedeg(f)=O(1/\epsilon)$ and there exists an XOS function that requires degree $\Omega(1/\epsilon)$.
\end{itemize}
This improves on previous approaches that all showed an upper bound of $O(1/\eps^2)$ for submodular \cite{CheraghchiKKL:12,FeldmanKV:13,FeldmanVondrak:13arxiv} and XOS \cite{FeldmanVondrak:13arxiv} functions. The best previous lower bound was $\Omega(1/\eps^{2/3})$ for monotone submodular functions \cite{FeldmanKV:13}.
Our techniques reveal new structural properties of submodular and XOS functions and the upper bounds lead to nearly optimal PAC learning algorithms for these classes of functions.


\end{abstract}
\thispagestyle{empty}
\newpage
\setcounter{page}{1}
\section{Introduction}
\label{sec:intro}
Analysis of the discrete Fourier transform of functions over the hypercube has a wide range of notable applications in theoretical computer science. It is also the object of significant research interest in its own right \cite{ODonnell14:book}. While most of this research has been devoted to Boolean-valued functions, many works analyze general real-valued functions (\eg \cite{Talagrand:94,DinurFKO:06}). Recently, the analysis of real-valued functions over the hypercube has also attracted significant attention due to applications in learning theory, property testing, differential privacy, algorithmic game theory and quantum complexity \cite{GHIM09,BalcanHarvey:12full,GuptaHRU:11,SV11,CheraghchiKKL:12,BadanidiyuruDFKNR:12,BalcanCIW:12,
RaskhodnikovaYaroslavtsev:13,FeldmanKV:13,FeldmanVondrak:13arxiv,FeldmanK14,BermanRY14,AaronsonA:14,BackursB:14}. Most of the Fourier-analytic techniques apply to real-valued functions as well but many new questions arise when one considers the richer structure of real-valued functions.

Our focus is on {\em structural properties} of two fundamental classes of real-valued functions: submodular and fractionally subadditive. Submodularity, a discrete analog of convexity, has played an essential role in combinatorial optimization~\cite{E70,L83,Q95,F97,FFI00} and, more recently, in algorithmic game theory and machine learning \cite{GKS05,LLN06,DS06,KGGK06,KSG08,Vondrak08}.
In algorithmic game theory, submodular functions have found application as {\em valuation functions}
with the property of diminishing returns \cite{LLN06,DS06,Vondrak08}. Along with submodular functions, fractionally subadditive functions have been studied in the algorithmic game theory context \cite{LLN06} (see Sec.~\ref{sec:prelims} for the definition). Feige showed that these functions have an additional characterization as a maximum of non-negative linear functions or XOS \cite{Feige:06}. Here we also show that the Rademacher complexity of a set of vectors that plays a fundamental role in statistical learning gives yet another equivalent way to define this class of functions. For comparison, we also discuss the class of self-bounding functions that contains both submodular and XOS functions and shares a number of properties with those classes such as dimension-free concentration of measure \cite{BoucheronLM:00}. Informally, a function $f:\zo^n \rightarrow \R$ is self-bounding if for every $x \in \zo^n$, $f(x)$ upper bounds the sum of all the $n$ marginal decreases in the value of the function at $x$.  We define these classes and their relationships in Section~\ref{sec:prelims}.

The primary property we consider is how well these functions can be approximated by {\em low-degree polynomials}, where the approximation is measured in $\ell_2$ norm over the uniform distribution $\U$ defined as $\|f-g\|_2= \sqrt{\E_\U[(f(x)-g(x))^2]}$. By the standard duality for the $\ell_2$ norm, approximability of $f$ by polynomials of degree $d$ is characterized by how much of $f$'s Fourier weight resides on coefficients of degree above $d$. Concentration of the Fourier spectrum on low-degree coefficients is one of the central and most well-studied properties in Fourier analysis and its applications. In particular, following the seminal work of Linial, Mansour and Nisan \cite{LinialMN:93}, a large number of learning algorithms over the uniform (and other) distributions relies crucially on approximation by low-degree polynomials (\eg \cite{KalaiKMS:08,KlivansS08,KaneKM13}).

Motivated by learning of submodular functions and its application in differential privacy in \cite{GuptaHRU:11}, Cheraghchi \etal \cite{CheraghchiKKL:12} proved that every submodular function\footnote{Here and below we normalize the function range to $[0,1]$.} can be $\eps$-approximated in $\ell_2$ norm by a polynomial of degree $O(1/\eps^2)$. Their proof is based on the analysis of the noise sensitivity of submodular functions, a standard tool from Fourier analysis for establishing low-degree spectral concentration. Subsequently, Feldman \etal proved the same upper bound of $O(1/\eps^2)$ using approximation of submodular functions by real-valued decision trees \cite{FeldmanKV:13}. They also gave a lower bound for learning that implies a lower bound of $\Omega(1/\eps^{2/3})$ on the degree necessary to $\eps$-approximate submodular functions.

Most recently, we considered the approximability of submodular and XOS functions by functions of few variables or {\em juntas} \cite{FeldmanVondrak:13arxiv}. We showed that submodular functions are $\eps$-approximated in $\ell_2$ by functions depending on $O(\frac{1}{\epsilon^2} \log \frac{1}{\epsilon})$ variables, while for XOS functions, a junta of size $2^{O(1/\epsilon^2)}$ suffices. In addition, we showed that submodular and XOS functions (in fact, all self-bounding functions) have constant total influence implying that they can be approximated by a polynomial of degree $O(1/\eps^2)$. These results have lead to substantially faster learning and testing algorithms for these classes of functions, most notably, a $2^{\tilde{O}(1/\epsilon^2)} \cdot n^2$ time PAC learning algorithm for submodular functions and a $2^{O(1/\epsilon^4)} \cdot n$ time PAC learning algorithm for XOS functions. Learning of submodular and XOS functions is also the main motivating application of this work.

\eat{
Clearly, a function on $\{0,1\}^n$ depending on at most $\tilde{O}({1}/{\epsilon^2})$ variables can be represented as a polynomial of degree $\tilde{O}({1}/{\epsilon^2})$. However, while submodular (even linear) functions require $\Omega({1}/{\epsilon^2})$ variables to approximate within an $\ell_2$-error of $\epsilon$, the bound in terms of low-degree approximation was not known to be tight.
We remark that even given this strong junta approximation result, approximating by a lower-degree polynomial is still meaningful: To describe a function of $\tilde{\Theta}(1/\epsilon^2)$ variables (or to learn it from examples), we need $2^{\tilde{\Theta}(1/\epsilon^2)}$ real coefficients. A lower-degree polynomial over $\tilde{\Theta}(1/\epsilon^2)$ variables can be described and also learned more efficiently.}

\subsection{Our Results}
\label{sec:our-results}
In this work, we investigate the degree that is necessary to approximate XOS and submodular functions in detail. For a real-valued function $f:\zon \rightarrow \R$ let $\aedeg(f)$ denote the smallest $d$ such that there is a polynomial $p$ of degree $d$ for which $\|f - p\|_2 \leq \epsilon$ and we refer to it as $(\ell_2,\eps)$-approximate degree of $f$.  The three known upper bounds on $(\ell_2,\eps)$-approximate degree of submodular functions are all $O(1/\eps^2)$ \cite{CheraghchiKKL:12,FeldmanKV:13,FeldmanVondrak:13arxiv}. The bounds are derived via three different approaches
suggesting that this might be the right answer. This bound also applies to XOS and self-bounding functions \cite{FeldmanVondrak:13arxiv} and the known lower bound of $\Omega(1/\eps^{2/3})$ also applies to all of these classes of functions \cite{FeldmanKV:13}.
Here we show that, in fact, the picture is substantially richer: each of these classes requires a different degree to approximate that corresponds to the increasing complexity of functions in these classes. We detail our bounds
\iffull below and also summarize them in Figure \ref{fig:table}.
\fi
\begin{itemize}
\item For any submodular function $f:\zo^n \rightarrow [0,1]$, $\aedeg(f)=O(\log (1/\epsilon)/\epsilon^{4/5})$. This is almost tight: we prove that even for very simple submodular functions of the form $f(x) = \min \{ \frac{2}{k} \sum_{i=1}^{k} x_i, 1 \}$ for some $k$, $\aedeg(f)=\Omega(1/\epsilon^{4/5})$.
\item For any XOS function $f:\zo^n \rightarrow [0,1]$, $\aedeg(f)=O(1/\epsilon)$. We also show that degree $\Omega(1/\epsilon)$ is necessary for XOS functions.
\end{itemize}
\iffull
\begin{figure}
\centering
$\begin{array}{|| c || c | c ||} \hline
\mbox{Class of functions} & \mbox{lower bound} & \mbox{upper bound} \\
\hline \hline
\mbox{linear} & 1 & 1 \\
\hline
\mbox{coverage} & \Omega(\log ({1}/{\epsilon})) \cite{FeldmanK14}  & O(\log ({1}/{\epsilon})) \  \cite{FeldmanK14} \\
\hline
\mbox{submodular} & \Omega(1/\epsilon^{4/5}) & O(1/\epsilon^{4/5} \cdot \log (1/\epsilon))  \\
\hline
\mbox{XOS} & \Omega(1/\epsilon)  & O(1/\epsilon) \\
\hline
\mbox{self-bounding} & \Omega(1/\epsilon^2) & O(1/\epsilon^2) \cite{FeldmanVondrak:13arxiv} \\
\hline
\end{array} $
\caption{\small Overview of low-degree approximations: bounds on $(\ell_2,\eps)$-approximate degree for a function with range $[0,1]$. \label{fig:table}}
\end{figure}
\fi

For comparison we show that the bounds above do not hold for the more general class of self-bounded functions (and, consequently, for functions with constant total influence). Namely, we show that there exists a self-bounding function $f:\zo^n \rightarrow [0,1]$, such that $\aedeg(f)=\Omega(1/\epsilon^2)$. This matches the upper bound in \cite{FeldmanVondrak:13arxiv}.
As an additional point of comparison, coverage functions, a subclass of submodular functions, can be approximated by polynomials of exponentially smaller $O(\log(1/\eps))$ degree \cite{FeldmanK14}. At the same time monotone functions and subadditive functions cannot be approximated by polynomials of dimension-free degree and require $\Omega(\sqrt{n})$ and $\Omega(n)$ degree, respectively, to approximate within a constant.

As a first application we show that the improved upper bound on $\aedeg$ of XOS functions leads to an upper bound of $2^{O(1/\epsilon)}$ on the size of junta sufficient to approximate an XOS function within $\ell_2$ error of $\eps$. This improves on the $2^{O(1/\epsilon^2)}$ upper bound and matches the lower bound of $2^{\Omega(1/\epsilon)}$ in \cite{FeldmanVondrak:13arxiv}.

\paragraph{Our techniques:}
It is easy to verify that previous approaches to proving upper bounds on $\aedeg$ cannot lead to upper bounds stronger than $1/\eps^2$ even in the case of submodular functions. For example, a bound on the total sum of squared influences $\infl^2(f)$ leads to $\aedeg(f) \leq \infl^2(f)/\eps^2$. However, $\infl^2(f) = 1$ even for the monotone submodular function $f(x) = x_1$.

The first step of both of our upper bounds is a spectral concentration bound based on the total {\em second-degree} influences. Namely, we consider the quantity $\sum_{i,j=1}^{n} \|\partial_{ij} f\|_2^2$, where $\partial_{ij} f$ is a second-degree discrete partial derivative of $f$. This quantity measures interactions between pairs of variables. It is particularly meaningful in the setting of submodular functions, where it measures the drop in marginal value of element $i$ due to the presence of $j$. That is, we always have $\partial_{ij} f \leq 0$ for submodular functions. We prove that for XOS functions the quantity $\sum_{i,j=1}^{n} \|\partial_{ij} f\|_2^2$ is at most a constant. This leads to an upper bound of $O(1/\epsilon)$ on $\aedeg(f)$, since $\sum_{i,j=1}^{n} \|\partial_{ij} f\|_2^2 \simeq 16 \cdot \sum_S |S|^2 \hat{f}^2(S)$. The proof of this bound is based on a careful analysis of contributions of the linear functions in the XOS representation. In particular, it also reveals that XOS functions satisfy a degree-two version of self-boundedness property: for all $x$, $\sum_{i,j: x_i=x_j=1} (\partial_{ij} f(x))^2 \leq 5 (f(x))^2$ (for comparison, self-bounding monotone functions satisfy $\sum_{i: x_i=1} \partial_{i} f(x) \leq f(x)$).


The upper bound above is optimal for XOS functions. To prove this we give an embedding of monotone DNF formulas into XOS functions. We then use the high noise sensitivity of Talagrand's random DNF \cite{mossel2002noise} to prove our lower bound on the low-degree spectral concentration.

For submodular functions, we use a different approach to obtain the stronger bound. We examine how the sum of second-degree influences $\sum_{i,j=1}^{n} \|\partial_{ij} f\|_2^2$ behaves when no individual influence is too large. The technical notion of ``large" that we use is the following ``threshold norm": $\|\partial_i f\|_T = \sup \{\alpha \geq 0: \Pr[|\partial_i f(x)| \geq \alpha] \geq \alpha^3 \}$. We prove that at most $O(\frac{1}{\epsilon} \log \frac{1}{\epsilon})$ partial derivatives can be large in the sense that $\|\partial_i f\|_T > \epsilon$. This result is a special case of almost-everywhere boundedness of almost all the partial derivatives of a submodular function that we show. Namely, the number of variables $i$ for which $\Pr[|\partial_i f(x)| \geq \alpha] \geq \delta$  is at most $O(\log(1/\delta)/\eps)$. To prove this result we rely on the ``boosting lemma" of Goemans and Vondrak \cite{GoemansVondrak:06}, also used in our recent work \cite{FeldmanVondrak:13arxiv}. (We note that an equivalent statement also appeared in \cite{KahnKalai07}.) Finally, we prove that for submodular functions with partial derivatives bounded by $\|\partial_i f\|_T \leq \epsilon$, we have $\sum_{i,j=1}^{n} \|\partial_{ij} f\|_2^2 = \tilde{O}(\sqrt{\epsilon})$. This leads to the upper bound of $\tilde{O}(1/\epsilon^{4/5})$.

As a warm-up to the upper bound for general submodular functions we also show a substantially simpler analysis for totally symmetric submodular functions (functions invariant under permutations of variables). In this case we avoid the logarithmic factor and get an $O(1/\epsilon^{4/5})$ upper bound. While the exponent of $\eps$ in our upper bound is quite unexpected it is actually optimal.
In particular, using direct estimation of spectral concentration we show that the simple function $f(x) = \min \{\frac{2}{k} \sum_{i=1}^{k} x_i, 1 \}$ requires degree $\Omega(1/\epsilon^{4/5})$ for $k = \Theta(1/\epsilon^{4/5})$. This function is monotone, totally symmetric, budget-additive and also can be viewed as a scaled rank function of a uniform matroid. Hence the lower bound applies to these subclasses of submodular functions as well. We remark that the weaker lower bound of $\Omega(1/\eps^{2/3})$ in \cite{FeldmanKV:13} was given for same function.

Finally, the lower bound of $\Omega(1/\eps^2)$ for self-bounding functions is based on an embedding of any Boolean function into a self-bounding Boolean function over $n + \log(n) + O(1)$ variables using the Hamming error-correcting code of distance 3.

\paragraph{Learning:}
The new structural results directly translate into improved learning algorithms using the techniques from \cite{FeldmanKV:13,FeldmanVondrak:13arxiv}. For brevity we describe the improvements for learning from random examples in the PAC model with $\ell_2$ error. Similar improvements can be obtained for agnostic learning and learning with value queries (which allow the learner to ask for the value of the function at any point). For both XOS and submodular functions we give a new lower bound which shows that our learning algorithms are close to optimal.
\begin{theorem}
\label{thm:submod-learning-intro}
There exists an algorithm $\A$ that given $\eps > 0$ and access to random uniform examples of a submodular XOS function $f:\zo^n \rightarrow [0,1]$, with probability at least $2/3$, outputs a function $h$, such that $\|f-h\|_2 \leq \epsilon$. Further, $\A$ runs in time $2^{\tilde{O}(1/\eps^{4/5})} \cdot n^2$ and uses $2^{\tilde{O}(1/\eps^{4/5})} \log n$ random examples.
\end{theorem}
The best previous algorithm for this task runs in time $2^{\tilde{O}(1/\eps^{2})} \cdot n^2$ and uses $2^{\tilde{O}(1/\eps^{2})} \log n$ random examples \cite{FeldmanVondrak:13arxiv}. We complement the new learning upper bound by a nearly tight information-theoretic lower bound of $2^{\Omega(1/\eps^{4/5})}$ examples (of value queries) for any PAC learning
\iffull
algorithm (see Thm.~\ref{thm:submod-PAC-hardness} for a formal statement).
\fi

The proof of the lower bound relies on the reduction in \cite{FeldmanKV:13}.
The improved polynomial approximation and junta size for XOS functions lead to the following improved PAC learning algorithm.
\begin{theorem}
\label{thm:XOS-learning-intro}
There exists an algorithm $\A$ that given $\eps > 0$ and access to random uniform examples of an XOS function $f:\zo^n \rightarrow [0,1]$, with probability at least $2/3$, outputs a function $h$, such that $\|f-h\|_2 \leq \epsilon$. Further, $\A$ runs in time $r^{O(1/\eps)} \cdot n$ and uses $r^{O(1/\eps)} \log n$ random examples, where $r = \min\{n, 2^{1/\eps}\}$.
\end{theorem}
The best previous upper bound was polynomial in $r^{O(1/\eps^2)}$ for $r = \min\{n, 2^{1/\eps^2}\}$ \cite{FeldmanVondrak:13arxiv}. We prove that any PAC algorithm for XOS functions requires $2^{\Omega(1/\eps)}$ examples\iffull (see Thm.~\ref{thm:xos-pac-hardness} for a formal statement)\fi. This upper bound is close to being tight when $n$ is subexponential in $2^{1/\eps}$. The lower bound is based on the embedding of monotone DNF into XOS functions that we used in the lower bound on $\aedeg$ together with the lower bound for learning monotone DNF of Blum \etal \cite{BlumBL:98}.
\iffull
Finally, using the Hamming code-based embedding we mentioned above we give a stronger lower bound of $2^{\Omega(1/\eps^2)}$ examples for any PAC learning algorithm for learning self-bounding functions.
\fi

\eat{
 is necessary to achieve an $\epsilon$-approximation in $\ell_2$. The example is based on Hadamard's code of distance $3$. It was already known that self-bounding functions can be approximated within $\ell_2$ error $\epsilon$ by polynomials of degree $O(1/\epsilon^2)$.
  the example is based on Talagrand's DNF formula.
}

\iffull
\paragraph{Organization:} Following the preliminaries we present the proofs of our main upper bounds: in Section \ref{sec:xos} for XOS functions and in Section \ref{sec:submod} for submodular functions. Applications to approximation of XOS functions by juntas and learning algorithms appear in Section \ref{sec:applications}. Details of lower bounds on spectral concentration and learning appear in Section \ref{sec:lower-bounds}. In Appendix \ref{sec:rademacher} we prove the equivalence of Rademacher complexity and XOS functions.
\fi

\subsection{Related Work}
\label{sec:related-work}
Analysis of functions on the Boolean hypercube has a long history with strong ties to combinatorics, probability, learning theory, cryptography and complexity theory (see \cite{ODonnell14:book}). One of the fundamental and most well-studied properties of Boolean functions is monotonicity. There is now a rich and detailed literature on the structure of the Fourier spectrum of monotone Boolean functions and their learnability over the uniform distribution \cite{KearnsLV:94,Talagrand:94,Talagrand:96,BshoutyTamon:96,BlumBL:98,mossel2002noise,Odonnell:03thesis,AmanoM06,Dachman-SoledLMSWW08,OWimmer:09,DachmanFTWW:15}. Starting with the work of Goldreich \etal \cite{GoldreichGLRS00} numerous works have also investigated testing of monotone functions over the Boolean hypercube. Submodularity is closely related to monotonicity: indeed a function is submodular if and only if its partial derivatives are monotone non-increasing. In addition, XOS functions which are monotone share structural similarities with monotone DNF formulas (we make this explicit in Section \ref{sec:lower-xos}). Hence our work is both inspired by the research on understanding of monotonicity over the Boolean hypercube and builds on techniques and results developed in that research. At the same time we are not aware of techniques closely related to those we use to prove our upper bounds for submodular and XOS functions having been used before.

\iffull
We now review some recent work on learning of submodular, XOS and related classes of real-valued functions.
Reconstruction of submodular functions up to some multiplicative factor (on every point) from value queries was first considered by Goemans \etal \cite{GHIM09}. They show a polynomial-time algorithm for reconstructing monotone submodular functions with $\tilde{O}(\sqrt{n})$-factor approximation and prove a nearly matching lower-bound. This was extended to the class of all subadditive functions in \cite{BadanidiyuruDFKNR:12} which studies small-size approximate representations of valuation functions (referred to as {\em sketches}).

\fi
Motivated by applications in economics, Balcan and Harvey initiated the study of learning submodular functions from random examples coming from an unknown distribution and introduced the PMAC learning model that requires a multiplicative approximation to the target function on most of the domain \cite{BalcanHarvey:12full}. They give an O($\sqrt{n}$)-factor PMAC learning algorithm and show an information-theoretic $\Omega(\sqrt[3]{n})$-factor impossibility result for submodular functions. Subsequently, Balcan \etal gave a distribution-independent PMAC learning algorithm for XOS functions that achieves an $\tilde{O}(\sqrt{n})$-approximation and showed that this is essentially optimal \cite{BalcanCIW:12}.


Learning of submodular functions with additive rather than multiplicative guarantees over the uniform distribution was first considered by Gupta \etal who were motivated by applications in private data release \cite{GuptaHRU:11}. \iffull
They show that submodular functions can be $\epsilon$-approximated by a collection of $n^{O(1/\epsilon^2)}$ $\epsilon^2$-Lipschitz submodular functions. Concentration properties imply that each $\epsilon^2$-Lipschitz submodular function can be $\epsilon$-approximated by a constant. This leads to
\fi
a learning algorithm running in time $n^{O(1/\epsilon^2)}$, which however requires value queries in order to build the collection. Using the upper bound of $O(1/\eps^2)$ on $\aedeg$ of submodular functions Cheraghchi \etal gave a $n^{O(1/\eps^2)}$ learning algorithm which uses only random examples and, in addition, works in the agnostic setting \cite{CheraghchiKKL:12}.
\iffull
Feldman \etal show that the decomposition from \cite{GuptaHRU:11} can be computed by a low-rank binary decision tree \cite{FeldmanKV:13}. They then show that this decision tree can then be pruned to obtain depth $O(1/\eps^2)$ decision tree that approximates a submodular function. This construction implies approximation by a $2^{O(1/\eps^2)}$-junta of degree $O(1/\eps^2)$.
\fi
They also show how approximation by a junta can be used to obtain a $2^{O(1/\eps^4)}$ PAC learning algorithm for submodular functions. Feldman \etal extend the results on noise sensitivity of submodular functions in \cite{CheraghchiKKL:12} to all self-bounding functions and show that they imply approximation within $\ell_1$ distance of $\eps$ by a polynomial of $O(\log(1/\eps)/\eps)$ degree and $2^{O(\log(1/\eps)/\eps)}$-junta \cite{FeldmanKV14}. Note that approximation in $\ell_2$ norm we give here is stronger and our lower bound for self-bounding functions implies that any approach that works for all self-bounding functions cannot improve on the $O(1/\eps^2)$ bound on $\aedeg$.

\iffull
Raskhodnikova and Yaroslavtsev consider learning and testing of submodular functions taking values in the range $\{0,1,\ldots,k\}$ (referred to as {\em pseudo-Boolean}) \cite{RaskhodnikovaYaroslavtsev:13}. The error of a hypothesis in their framework is the probability that the hypothesis disagrees with the unknown function. They show that pseudo-Boolean submodular functions can be expressed as $2k$-DNF and thus obtain a $\poly(n) \cdot k^{O(k \log{k/\eps})}$-time PAC learning algorithm using value queries. In a subsequent work, Blais \etal prove existence of a junta of size $(k \log(1/\eps))^{O(k)}$ and use it to give an algorithm for testing submodularity using $(k \log(1/\eps))^{\tilde{O}(k)}$ value queries \cite{BlaisOSY:13manu}.
\fi

\section{Preliminaries} \label{sec:prelims}


Let us define submodular, fractionally subadditive and subadditive functions. These classes are well known in combinatorial optimization and there has been a lot of recent interest in these functions in algorithmic game theory, due to their expressive power as {\em valuations} of self-interested agents.

\begin{definition}
A set function $f:2^N \rightarrow \RR$ is
\begin{compactitem}
 \item monotone, if $f(A) \leq f(B)$ for all $A \subseteq B \subseteq N$.
 \item submodular, if $f(A \cup B) + f(A \cap B) \leq f(A) + f(B)$ for all $A,B \subseteq N$.
 \item subadditive, if $f(A \cup B) \leq f(A) + f(B)$ for all $A \subseteq B \subseteq N$.
 \item fractionally subadditive, if $f(A) \leq \sum \beta_i f(B_i)$ whenever $\beta_i \geq 0$
 and $\sum_{i:a \in B_i} \beta_i \geq 1 \ \forall a \in A$.
\end{compactitem}
\end{definition}

We identify functions on $\zo^n$ with set functions on $N = [n]$ in a natural way. By $\bf 0$ and $\b1$, we denote the all-zeroes and all-ones vectors in $\{0,1\}^n$ respectively. Submodular functions are not necessarily nonnegative, but in many applications (especially when considering multiplicative approximations), this is a natural assumption. All our approximations are shift-invariant and hence also apply to submodular functions with range $[-1/2,1/2]$ (and can also be scaled in a straightforward way). Fractionally subadditive functions are nonnegative by definition (by considering $A=B_1, \beta_1 > 1$) and satisfy $f({\bf 0}) = 0$ (by considering $A = B_1 = \emptyset, \beta_1 = 0$). There is an equivalent definition known as ``XOS" or maximum of non-negative linear functions \cite{Feige:06}:
$f(x) = \max_{c \in C} \sum_{i=1}^{n} w_{ci} x_i.$
Here, $w_{ci} \geq 0$ are nonnegative weights. 
This class includes all (nonnegative) monotone submodular functions such that $f({\bf 0}) = 0$ (but does not contain non-monotone functions).
\iffull In Appendix \ref{sec:rademacher}
\fi
we show that Rademacher complexity of a set of vectors, a powerful and well-studied tool in statistical learning theory \cite{KoltchinskiiPanchenko:00,BartlettMendelson:02}, gives an equivalent way to define XOS functions.
We also show that the class of monotone self-bounding functions is stricly broader than than of XOS functions.

A broader class is that of {\em self-bounding functions}. Self-bounding functions were defined by Boucheron, Lugosi and Massart
\cite{BoucheronLM:00} and further generalized by McDiarmid and Reed \cite{MR06} as a unifying class of functions that enjoy strong concentration properties. Here, we define self-bounding functions in the special case of $\zo^n$ as follows.
A function $f:\zo^n \rightarrow \RR$ is called $a$-self-bounding, if $f$ is $1$-Lipschitz and for all $x \in \zo^n$,
$$ \sum_{i=1}^{n} (f(x) - f(x \oplus e_i))_+ \leq a f(x),$$
where $x \oplus e_i$ is $x$ with $i$-th bit flipped and $(\alpha)_+$ denotes $\max\{0,\alpha\}$. 
The $1$-Lipschitz condition does not play a role in this paper, as we normalize functions to have values in the $[0,1]$ range.
Self-bounding functions subsume fractionally subadditive functions, and 2-self-bounding functions subsume (possibly non-monotone) submodular functions.  See \cite{FeldmanVondrak:13arxiv} for a more detailed discussion of these classes of functions.

The $\ell_1$ and $\ell_2$-norms of $f:\zon\rightarrow \RR$ are defined by $\|f\|_1 =  \E_{x \sim \U} [|f(x)|]$ and $\|f\|_2 =  (\E_{x \sim \U} [f(x)^2])^{1/2}$, respectively, where $\U$ is the uniform distribution.
\iffull
\begin{definition}[Discrete derivatives]
\fi
For $x \in \zon$, $b \in \zo$ and $i \in n$, let $x_{i\leftarrow b}$ denote the vector in $\zo^n$ that equals $x$  with $i$-th coordinate set to $b$. For a function $f:\zon \rightarrow \RR$ and index $i \in [n]$ we define
$\partial_i f(x) = f(x_{i\leftarrow 1}) - f(x_{i\leftarrow 0})$.
We also define $\partial_{i,j} f(x) = \partial_i \partial_j f(x)$.
\iffull
\end{definition}
\fi

\iffull
A function is monotone (non-decreasing) if and only if for all $i\in [n]$ and $x\in \zo^n$, $\partial_i f(x) \geq 0$.
For a submodular function, $\partial_{i,j} f(x) \leq 0$, by considering the submodularity condition for $x_{i \leftarrow 0, j \leftarrow 0}$, $x_{i \leftarrow 0, j \leftarrow 1}$, $x_{i \leftarrow 1, j \leftarrow 0}$, and  $x_{i \leftarrow 1, j \leftarrow 1}$.
\fi

\iffull
\smallskip
\noindent{\bf Absolute error vs.~error relative to norm:}
In our results, we typically assume that the values of $f(x)$ are in a bounded interval $[0,1]$,
and our goal is to learn $f$ with an additive error of $\epsilon$. Some prior work considered an error relative to the norm of $f$, for example at most $\epsilon \|f\|_2$ \cite{CheraghchiKKL:12}. In fact, it is known that for a non-negative submodular, XOS or self-bounding function $f$, $\|f\|_2 = \Omega(\|f\|_\infty)$ \cite{Feige:06,FeigeMV:07,FeldmanKV14} and hence this does not make much difference. If we scale $f(x)$ by $\frac{1}{4 \|f\|_2}$, we obtain a function with values in $[0,1]$ and learning the original function within an additive error of $\epsilon \|f\|_2$ is equivalent to learning the scaled function within an error of $\epsilon/4$.
\fi

\smallskip
\noindent{\bf Fourier Analysis:}
\label{sec:lowsens-prelims}
We rely on the standard Fourier transform representation of real-valued functions over $\zon$ as linear combinations of parity functions.
For $S \subseteq [n]$, the parity function $\chi_S:\zon \rightarrow \on$ is defined by
$ \chi_S(x) = (-1)^{\sum_{i \in S} x_i}$.
The Fourier expansion of $f$ is given by $f(x) = \sum_{S \subseteq [n]} \hat{f}(S) \chi_S(x)$. The {\em Fourier degree} of $f$ is the largest $|S|$ such that $\hat{f}(S) \neq 0$. Note that Fourier degree of $f$ is exactly the polynomial degree of $f$ when viewed over $\on^n$ instead of $\zon$ and therefore it is also equal to the polynomial degree of $f$ over $\zon$. 

For degree $d$, let $W^{d}(f) = \sum_{S \subseteq [n],\ |S| = d} (\hat{f}(S))^2$ and $W^{>d}(f) = \sum_{i > d} W^i(f)$.
For any function $f$, Parseval's identity states that $\|f\|_2^2 = \sum_{S \subseteq [n]} (\hat{f}(S))^2$.
This implies that the degree $d$ polynomial closest in $\ell_2$ distance to $f$ is precisely $p(x) = \sum_{S \subseteq [n],\ |S| \leq d} \hat{f}(S) \chi_S(x)$ and $\|f - p\|_2 = \sqrt{W^{>d}(f)}$. In other words, $\aedeg(f) = d$ if and only if $d$ is the smallest such that $W^{>d}(f) \leq \eps^2$.
\iffull

\fi
Observe that:
$\partial_i f(x) = -2 \sum_{S \ni i}\hat{f}(S)\chi_{S\setminus\{i\}}(x)$, and
$\partial_{i,j} f(x) = 4 \sum_{S \ni i,j}\hat{f}(S)\chi_{S\setminus\{i,j\}}(x)$.


\section{Degree $O(1/\epsilon)$ approximation for XOS functions}
\label{sec:xos}

In this section, we consider XOS functions $f:\zo^n \rightarrow \RR_+$,
$ f(x) = \max_{c \in C} \sum_{i=1}^{n} w_{ci} x_i,$
where, $w_{ci} \geq 0$ are nonnegative weights. We call each $c \in C$ a {\em clause} of the XOS function.

We recall that XOS functions, and more generally self-bounding functions, satisfy the following inequality for each $x \in \zo^n$:
$ \sum_{i=1}^{n} (f(x) - f(x \oplus e_i))_+  \leq f(x).$
In particular, for XOS functions (which are monotone), this can be written as
\begin{equation}
\label{eq:self-bound}
 \sum_{i: x_i=1} \partial_i f(x) \leq f(x).
\end{equation}
This leads to a bound of the form $\sum_S |S| \hat{f}^2(S) = O(\|f\|^2_2)$, which implies that degree $O(1/\epsilon^2)$ is sufficient to approximate XOS functions within $\ell_2$-error $\epsilon$. Here, we aim to improve the degree bound from $O(1/\epsilon^2)$ to $O(1/\epsilon)$. For this purpose, we seek a ``second-degree variant" of inequality (\ref{eq:self-bound}), using the second-degree derivatives $$\partial_{ij} f(x) = f(x_{i \leftarrow 1, j \leftarrow 1}) - f(x_{i \leftarrow 1, j \leftarrow 0}) - f(x_{i \leftarrow 0, j \leftarrow 1}) + f(x_{i \leftarrow 0, j \leftarrow 0}).$$
(For $i=j$, we define $\partial_{ii} f(x) = 0$.)
Our plan is to use these expressions as follows.

\begin{lemma}
\label{lem:degree-2-bound}
For any function $f:\zo^n \rightarrow \RR_+$ and any $1 \leq k \leq n$,
$$ \sum_{|S| > k} \hat{f}^2(S) \leq \frac{1}{16k^2} \sum_{i,j=1}^{n} \| \partial_{ij} f \|_2^2.$$
\end{lemma}
\begin{proof}
For every pair $i \neq j \in [n]$, we have
$ \| \partial_{ij} f \|_2^2 = 16 \sum_{S: i,j \in S} \hat{f}^2 (S).$
Summing up over all choices of $i \neq j$, each set $S$ appears $|S| (|S|-1)$ times:
$$ \sum_{i \neq j} \| \partial_{ij} f \|_2^2 = 16 \sum_{S \subseteq [n]} |S| (|S|-1) \hat{f}^2(S).$$
Therefore, we obtain
$ \sum_{i,j=1}^{n} \| \partial_{ij} f \|_2^2 \geq 16 \sum_{S \subseteq [n]} |S| (|S|-1) \hat{f}^2(S) \geq 16 k^2 \sum_{|S|>k} \hat{f}^2(S).$
\end{proof}

Our goal in the following is to bound the expression $\sum_{i,j=1}^{n} \| \partial_{ij} f \|_2^2$.
First we prove the following.

\begin{lemma}
\label{lem:square-self-bound}
For an XOS function $f:\zo^n \rightarrow \RR_+$ and any $x \in \zo^n$,
\begin{equation}
\label{eq:square-self-bound}
 \sum_{i,j: x_i=x_j=1} (\partial_{ij} f(x))^2 \leq 5 (f(x))^2.
\end{equation}
\end{lemma}


\medskip

\begin{proof}
Let $S$ denote the set of coordinates such that $x_i = 1$. Let $c \in C$ be a clause that achieves the maximum, defining $f(x) = \sum_{j \in S} w_{cj}$ (if there are multiple such clauses, fix one arbitrarily). Fix any $i \in S$ and define $c'_i \in C$ to be a clause achieving the maximum that defines $f(x_{i \leftarrow 0}) = \sum_{j \in S \setminus \{i\}} w_{c_i 'j}$. Fix another $j \in S$. We claim the following bounds:
\begin{equation}
\label{eq:square-bound}
 -\min \{ w_{ci} + w_{cj}, w_{c'_i j} \} \leq \partial_{ij} f(x) \leq \min \{ w_{ci}, w_{cj} \}.
\end{equation}

First, assume that $\partial_{ij} f(x) > 0$. Since $f$ is monotone, we have $\partial_{ij} f(x) \leq \min \{ \partial_i f(x), \partial_j f(x) \}$. Since $c$ is the clause defining $f(x)$, $f(x)$ cannot decrease by more than $w_{ci}$ when flipping $x_i$ from $1$ to $0$. Similarly, $f(x)$ cannot decrease by more than $w_{cj}$ when flipping $x_j$ from $1$ to $0$. Therefore, $\partial_{ij} f(x) \leq \min \{ w_{ci}, w_{cj} \}$.

Second, assume that $\partial_{ij} f(x) < 0$. Here we have $\partial_{ij} f(x) \geq -\min \{ \partial_i f(x_{j \leftarrow 0}), \partial_j f(x_{i \leftarrow 0}) \}$. Recall that after flipping $x_i$ from $1$ to $0$, $c'_i$ is a maximizing clause, and therefore $\partial_j f(x_{i \leftarrow 0}) = f(x_{i \leftarrow 0, j \leftarrow 1}) - f(x_{i \leftarrow 0, j \leftarrow 0}) \leq w_{c'_i j}$. To bound $\partial_i f(x_{j \leftarrow 0})$, we use the following (by monotonicity): $\partial_i f(x_{j \leftarrow 0}) = f(x_{i \leftarrow 1, j \leftarrow 0}) - f(x_{i \leftarrow 0, j \leftarrow 0}) \leq f(x_{i \leftarrow 1, j \leftarrow 1}) - f(x_{i \leftarrow 0, j \leftarrow 0}) \leq w_{ci} + w_{cj}$, using the fact that $c$ is a maximizing clause for $f(x)$. (We remark that although this seems like a weak bound, it could be actually tight.)  This proves (\ref{eq:square-bound}).

Next, we sum up over all pairs of coordinates $i,j \in S$. Note that $c \in C$ is fixed before choosing $i,j$, and we can assume for convenience that the coordinates are ordered so that $i \leq j$ implies $w_{ci} \leq w_{cj}$. We have
\begin{eqnarray*}
 \sum_{i,j \in S} (\partial_{ij} f(x))^2 & = & \sum_{i,j \in S: \partial_{ij} f(x) > 0} (\partial_{ij} f(x))^2 +  2 \sum_{i,j \in S: i>j, \partial_{ij} f(x) < 0} (\partial_{ij} f(x))^2 \\
& \leq & \sum_{i,j \in S} w_{ci} w_{cj} + 2 \sum_{i,j \in S, i > j} (w_{ci} + w_{cj}) w_{c'_i j} \\
& \leq & \sum_{i,j \in S} w_{ci} w_{cj} + 4 \sum_{i,j \in S, i > j} w_{ci} w_{c'_i j} \\
& = & \left( \sum_{i \in S} w_{ci} \right)^2 + 4 \sum_{i \in S} \left( w_{ci} \sum_{j \in S, j<i} w_{c'_i j} \right) \\
& \leq & (f(x))^2 + 4 (f(x))^2 = 5 (f(x))^2
\end{eqnarray*}
since $\sum_{j \in S} w_{c' j} \leq f(x)$ for every clause $c' \in C$.
\end{proof}

\begin{lemma}
\label{lem:exp-square-bound}
For any XOS function $f:\zo^n \rightarrow \RR_+$,
$$ \sum_{i,j=1}^{n} \| \partial_{ij} f \|_2^2 \leq 20  \| f \|_2^2.$$
\end{lemma}

\begin{proof}
Since all norms here are over the uniform distribution, we have
$ \| \partial_{ij} f \|_2^2 = \frac{1}{2^n} \sum_{x \in \zo^n} (\partial_{ij} f(x))^2.$
Note that Lemma~\ref{lem:square-self-bound} counts only the contributions from points such that $x_i = x_j = 1$.
However, $\partial_{ij} f(x)$ does not depend on the values of $x_i$ and $x_j$. Therefore, we can write equivalently
$$ \| \partial_{ij} f \|_2^2 = \frac{4}{2^n} \sum_{x \in \zo^n: x_i=x_j=1} (\partial_{ij} f(x))^2.$$
Summing up over all $i,j$ and switching the sums, we get
$$ \sum_{i,j=1}^{n} \| \partial_{ij} f \|_2^2 = \frac{4}{2^n} \sum_{i,j=1}^{n} \sum_{x \in \zo^n: x_i=x_j=1} (\partial_{ij} f(x))^2
 = \frac{4}{2^n} \sum_{x \in \zo^n} \sum_{i,j: x_i=x_j=1} (\partial_{ij} f(x))^2.$$
Now, we can apply Lemma~\ref{lem:square-self-bound} to conclude that
$ \sum_{i,j=1}^{n} \| \partial_{ij} f \|_2^2  \leq \frac{4}{2^n} \sum_{x \in \zo^n} 5 (f(x))^2 = 20 \| f \|_2^2.$
\end{proof}

We can conclude as follows.

\begin{corollary}
\label{cor:XOS-degree-upper}
For any XOS function $f:\zo^n \rightarrow [0,1]$, there is a polynomial $p$ of degree $O(1/\epsilon)$ such that
$ \| f - p \|_2 \leq \epsilon.$
\end{corollary}

\begin{proof}
By Lemma~\ref{lem:exp-square-bound}, we have $\sum_{i,j=1}^{n} \| \partial_{ij} f \|_2^2 \leq 20$, since $\| f \|_2 \leq 1$ here.
Therefore, applying Lemma~\ref{lem:degree-2-bound},
$ \sum_{|S|>k} \hat{f}^2(S) \leq \frac{5}{4k^2}.$
We choose $k =\sqrt{5} / (2\epsilon)$, which ensures that $\sum_{|S|>k} \hat{f}^2(S) \leq \epsilon^2$ and therefore the polynomial consisting of all terms up to degree $k$ approximates $f$ within $\ell_2$-error $\epsilon$.
\end{proof}


\section{Degree $\tilde{O}(1/\epsilon^{4/5})$ approximation for submodular functions}
\label{sec:submod}

In this section, we show that the $O(1/\epsilon)$ degree approximation for XOS functions can be improved to $\tilde{O}(1/\epsilon^{4/5})$ for submodular functions. Interestingly, $1/\epsilon^{4/5}$ turns out to be the right answer for submodular functions (ignoring logarithmic factors).

We build on the technique of bounding $\sum_{i,j} \| \partial_{ij} f \|_2^2$, which in the case of submodular functions seems particularly appropriate since we know that $\partial_{ij} f(x) \leq 0$ for every $x \in \zo^n$, which simplifies certain expressions. However, Lemma~\ref{lem:exp-square-bound} itself cannot be improved to a sub-constant bound --- it is easy to see that $\sum_{i,j}  \| \partial_{ij} f \|_2^2$ could be at least $\| f \|^2_2$ for a submodular function (e.g., $f(x) = 1 - (1-x_1) (1-x_2)$). However, as we show below this can happen only when some variables have a very large influence. Our goal is to handle such variables separately and prove that under the assumption of low influences, the quantity $\sum_{i,j} \| \partial_{ij} f \|_2^2$ cannot be large.

Once we can control the influences of individual variables (for now imagine that we can control $\| \partial_i f \|_2$),
we use the following way of bounding the sum of second partial derivatives.

\begin{lemma}
\label{lem:square-root-bound}
For any submodular function $f:\zo^n \rightarrow \RR$, any coordinate $i$ and a subset of coordinates $A$,
$$ \sum_{j \in A} \|\partial_{ij} f\|_1 \leq 2 \sqrt{|A|} \ \| \partial_i f\|_2.$$
\end{lemma}

Note the improvement from $2 |A| \| \partial_i f\|_1$ (which is trivial) to $2 \sqrt{|A|} \| \partial_i f\|_2$ on the right-hand-side. 

\medskip

\begin{proof}
Since $f$ is submodular, we have $\partial_{ij} f(x) \leq 0$, and
\begin{eqnarray*}
\sum_{j \in A} \|\partial_{ij} f\|_1 = \sum_{j \in A} \E_{x \sim \U}[ -\partial_{ij} f(x)]
 = 2 \sum_{j \in A} \E_{x \sim \U}[ (-1)^{x_j} \partial_i f(x)] = 2 \cdot \E_{x \sim \U}[ \partial_i f(x) g(x) ]  \leq 2 \| \partial_i f \|_2 \| g \|_2
\end{eqnarray*}
where $g(x) = \sum_{j \in A} (-1)^{x_j}$ and we used the Cauchy-Schwartz inequality at the end.
It is easy to check that
$ \|g\|_2 = \sqrt{|A|} $ which proves the lemma.
\end{proof}

\smallskip

First, let us sketch how this argument leads to an $O(1/\epsilon^{4/5})$ bound in the case of {\em totally symmetric} submodular functions,  to illustrate some of the ideas employed in the general case.

\paragraph{Totally symmetric submodular functions.}
Let us assume that $f:\zo^n \rightarrow [0,1]$ is totally symmetric in the sense that $f(x)$ depends only on $\sum_{i=1}^{n} x_i$.
Note that such a function is simply a concave function of $\sum_{i=1}^{n} x_i$.
First, we observe that the influences of individual variables in this case cannot be too large.

\begin{lemma}
For any totally symmetric submodular function $f:\zo^n \rightarrow [0,1]$ and $x \in \zo^n$ such that $\frac{n}{3} \leq \sum_{i=1}^{n} x_i \leq \frac{2n}{3}$, we have $|\partial_i f(x)| \leq \frac{3}{n}$ for all $i \in [n]$.
\end{lemma}

\begin{proof}
Assume that $\partial_i f(x) > \frac{3}{n}$ (the opposite case is similar). Since the function is totally symmetric, we actually have $\partial_j f(x) > \frac{3}{n}$ for every $j \in [n]$. Also, $\sum_{i=1}^{n} x_i \geq \frac{n}{3}$. By submodularity, $f(x) - f(0) \geq \sum_{j: x_j=1} \partial_j f(x) > \frac{3}{n} \sum_{j=1}^{n} x_j \geq 1$. This contradicts the fact that the values of $f(x)$ are in $[0,1]$.
\end{proof}

\iffull
To simplify the analysis, let us assume that in fact, $|\partial_i f(x)| = O(\frac{1}{n})$ for all $i \in [n]$ and all $x \in \zo^n$. This can be accomplished by modifying the function in the regions where $\sum_{i=1}^{n} x_i < \frac{n}{3}$ or $>\frac{2n}{3}$ in such a way that $\partial_i f(x)$ is constant in each region. For example, if $t$ is maximum such that $\partial_i f(x') > \frac{3}{n}$ for $\sum_{i=1}^{n} x'_i = t$, let $f(x') = F$ for this point $x'$ (and we know that $t = \sum_{i=1}^{n} x'_i < \frac{n}{3}$). We can set $f(x) = F - \frac{3}{n} (t - \sum_{i=1}^{n} x_i)$ whenever $\sum_{i=1}^{n} < t$. Similarly, we adjust the function for $\sum_{i=1}^{n} x_i > \frac{2n}{3}$. These are sets of small measure (under the uniform distribution) and so any approximation of the modified function also works well for the original function. In the following, we assume that $|\partial_i f(x)| = O(\frac{1}{n})$ everywhere. Now we can show the following bound.
\fi

\begin{lemma}
If $|\partial_i f(x)| = O(\frac{1}{n})$ for all $i \in [n]$ and $x \in \zo^n$, then
$ \sum_{i,j=1}^{n} \| \partial_{ij} f \|_2^2 = O\left(\frac{1}{\sqrt{n}} \right).$
\end{lemma}

\begin{proof}
Note that the assumption on partial derivatives also implies that $|\partial_{ij} f(x)| = O(\frac{1}{n})$.
We estimate $ \sum_{i,j=1}^{n} \| \partial_{ij} f \|_2^2$ as follows:
\begin{eqnarray*}
\sum_{i,j=1}^{n} \| \partial_{ij} f \|_2^2 =  O\left(\frac{1}{n} \right) \sum_{i,j=1}^{n} \| \partial_{ij} f \|_1
= O\left(\frac{1}{\sqrt{n}} \right) \sum_{i=1}^{n} \| \partial_{i} f \|_2
\end{eqnarray*}
using Lemma~\ref{lem:square-root-bound} with $A = [n]$. Since we assume that $|\partial_i f(x)| = O(\frac{1}{n})$, it follows that
$\| \partial_i f \|_2 = O(\frac{1}{n})$ for all $i \in [n]$ which proves the lemma.
\end{proof}

\smallskip

Now we can apply the method of bounding the Fourier tail above a certain level using Lemma~\ref{lem:degree-2-bound}:
$$ \sum_{|S|>k} \hat{f}^2(S) \leq \frac{16}{k^2} \sum_{i,j=1}^{n} \| \partial_{ij} f\|^2_2 = O\left( \frac{1}{k^2 \sqrt{n}} \right).$$
We choose $k = 1 / (\epsilon n^{1/4})$ in order to make the Fourier tail bounded by $O(\epsilon^2)$ as it should be.
Finally, note that if $n \leq 1 / \epsilon^{4/5}$, we can take trivially a polynomial of degree $n$. Therefore, the non-trivial case
is when $n > 1 / \epsilon^{4/5}$ and then we have $k = 1 / (\epsilon n^{1/4}) \leq 1 / \epsilon^{4/5}$. This proves that degree $O(1/\epsilon^{4/5})$ is sufficient for totally symmetric submodular functions.

\paragraph{General submodular functions.}
Let us turn now to the case of general submodular functions. The main complication here is that some variables can have large influences and we need to handle those separately. The main technical lemma here is that there cannot be too many variables of large influence, measured in a suitable way.
\iffull
The most technical part of the proof is to prove that there cannot be too many variables of large influence, and the influences decay relatively fast as we consider more variables. We also have to define ``influence" in a suitable way.
\fi
We denote by $\mu_p$ a product distribution on $\zo^n$ such that $\Pr_{x \sim \mu_p}[x_i=1] = p$ for each $i \in [n]$. We prove the following.

\begin{lemma}
\label{lem:submod-large-derivatives}
Let $f:\zo^n \rightarrow [0,1]$ be a submodular function, $0 < \delta,\epsilon < 1$, and let
$$ J(\epsilon,\delta) = \{ i \in [n]:  \Pr_{x \sim \mu_{1/2}}[| \partial_i f(x) | \geq \epsilon] \geq \delta \}.$$
Then $|J(\epsilon,\delta)| = O(\frac{1}{\epsilon} \log \frac{1}{\delta})$.
\end{lemma}

We prove this using the ``boosting lemma" of \cite{GoemansVondrak:06} (which was already used for the purpose of approximating submodular functions by juntas in \cite{FeldmanVondrak:13arxiv}).

\medskip

{\bf Boosting Lemma.}
{\em Let $\cF \subseteq \zo^X$ be down-monotone (if $x \in \cF$ and $y \leq x$ coordinate-wise, then $y \in \cF$).
For $p \in (0,1)$, define $\sigma_p = \Pr_{x \sim \mu_p}[x \in {\cal F}]$.
Then
$ \sigma_p = (1-p)^{\phi(p)} $
where $\phi(p)$ is a non-decreasing function for $p \in (0,1)$.
}

\medskip

\begin{proof}[of Lemma~\ref{lem:submod-large-derivatives}]
Let
\begin{itemize}
\item $J^+(\epsilon,\delta) = \{ i \in [n]:  \Pr_{x \sim \mu_{1/2}}[\partial_i f(x)  \geq \epsilon] \geq \delta/2 \}.$
\item $J^-(\epsilon,\delta) = \{ i \in [n]:  \Pr_{x \sim \mu_{1/2}}[\partial_i f(x)  \leq -\epsilon] \geq \delta/2 \}.$
\end{itemize}
We have $J(\epsilon,\delta) \subseteq J^+(\epsilon,\delta) \cup J^-(\epsilon,\delta)$. Hence it is enough to bound
$|J^+(\epsilon,\delta)|$; the same bound on $|J^-(\epsilon,\delta)|$ follows by considering the function $\bar{f}(x) = f(\b1 - x)$.

For each $j \in [n]$, define
$$\cF^+_j = \{ x \in \zo^n: \partial_j f(x) \geq \epsilon \}.$$
By submodularity, this set is down-monotone.
By assumption, we have $\Pr_{x \sim \mu_{1/2}}[x \in \cF^+_j] \geq \delta/2$ for $j \in J^+(\epsilon,\delta)$.
Using the terminology of the boosting lemma, we have $\sigma_{1/2} = (1/2)^{\phi(1/2)}$ where $\phi(1/2) \leq \log_2 (2/\delta)$.
We define $q = 1 - (1/2)^{1/\log_2 (2/\delta)} \leq 1/2$. By the boosting lemma \cite{GoemansVondrak:06}, we have
$$ \Pr_{x \sim \mu_q}[x \in \cF^+_j] = (1-q)^{\phi(q)} \geq (1-q)^{\log_2 (2/\delta)} = \frac{1}{2} $$
for each $j \in J^+(\epsilon,\delta)$.  We also have $\Pr_{x \sim \mu_q}[x_j=1] = q$.
Note that $x_j=1$ and $x \in \cF^+_j$ are independent events,
since $x \in \cF^+_j$ depends only on $\partial_j f(x)$ and this is independent of $x_j$. Therefore,
$$ \Pr_{x \sim \mu_q}[x_j=1 \ \& \ x \in \cF^+_j] \geq \frac{q}{2} $$
for each $j \in J^+(\epsilon,\delta)$. Let $L(x) = \{j: x_j=1 \ \& \ x \in \cF^+_j \}$. We have
$$ \E_{x \sim \mu_q}[|L(x)|] \geq \E_{x \sim \mu_q}[| \{ j \in J^+(\epsilon,\delta): x_j=1 \ \& \ x \in \cF^+_j\} |]
 \geq \frac{q}{2} |J^+(\epsilon,\delta)|.$$
On the other hand, denoting by $\b1_{S}$ the indicator vector of $S$, for each $j \in L(x)$, we have $\partial_j f(\b1_{L(x)}) \geq \epsilon$ and therefore
$$ \epsilon |L(x)| \leq \sum_{j \in L(x)} \partial_j f(\b1_{L(x)}) \leq f(\b1_{L(x)}) \leq 1 $$
where we used submodularity in the second inequality.
This means that $|L(x)| \leq 1 / \epsilon$ with probability $1$. Therefore, we have $|J^+(\epsilon,\delta)| \leq \frac{2}{\epsilon q}$.
Recall that $q = 1 - (1/2)^{1/\log_2 (2/\delta)} \geq \frac{1}{2 \log_2 (2/\delta)}$ (using $\delta < 1$) which means
$ |J^+(\epsilon, \delta)| \leq \frac{4}{\epsilon} \log_2 \frac{2}{\delta}.$
\end{proof}

We use Lemma~\ref{lem:submod-large-derivatives} for two purposes. First, it allows us to take out a small set of variables $L$ whose derivatives can be large with large probability. Conditioned on these variables, we get an ``almost $\epsilon$-Lipschitz" function, for which using Lemma~\ref{lem:submod-large-derivatives} again allows us to prove an improved bound on $\sum_{i,j \notin L} \| \partial_{ij} f \|_2^2$.

We introduce the following notation (the ``threshold norm").
In the following, all probabilities and expectations are over the uniform distribution ($x \sim {\cal U}$).

\begin{definition}
For a function $f:\zo^n \rightarrow \R$, we define
$$ \| f \|_T  = \sup \{ \alpha: \Pr_{x}[|f(x)| \geq \alpha] \geq \alpha^3 \}.$$
\end{definition}

\iffull
We remark that $\|f\|_T$ is not really a norm --- it is not linear under scalar multiplication. In fact $\|f\|_T$ is never more than $1$. The choice of $\alpha^3$ is somewhat arbitrary here. The notation $\|f\|_T$ is convenient for our proof but in general we do not attribute any significance to it. Lemma~\ref{lem:submod-large-derivatives} (with $\delta = \epsilon^3$) implies the following.
\fi

\begin{corollary}
\label{cor:threshold-norm}
For a submodular function $f:\zo^n \rightarrow [0,1]$, the number of coordinates with $\| \partial_i f \|_T \geq \epsilon$ is at most $O(\frac{1}{\epsilon} \log \frac{1}{\epsilon})$.
\end{corollary}

\iffull
We also have the following useful property (which we apply to $h=\partial_i f$).

\begin{lemma}
\label{lem:2vsT}
For any $h:\zo^n \rightarrow [-1,1]$,
$\|h\|_2 \leq \sqrt{2} \|h\|_T.$
\end{lemma}

\begin{proof}
Suppose that $\|h\|_T = \eta$ and note that $0 \leq \eta \leq 1$. For every $\alpha > \eta$, we have by definition $\Pr[\|h(x)\| \geq \alpha] < \alpha^3]$. Consequently
$$\|h\|_2^2 = \E[(h(x))^2] \leq \alpha^2 \cdot \Pr[|h(x)| \leq \alpha] + 1 \cdot \Pr[|h(x)| > \alpha] \leq \alpha^2 + \alpha^3.$$
 Since this holds for every $\alpha > \eta$, we also have $\|h\|_2^2 \leq \eta^2 + \eta^3 \leq 2 \eta^2$.
\end{proof}
\fi

The following is our main bound on the quantity $\sum_{i,j} \| \partial_{ij} f \|_2^2$.

\begin{lemma}
\label{lem:submod-Lipschitz-bound}
Let $f:\zo^n \rightarrow [0,1]$ be a submodular function such that $\| \partial_i f\|_T \leq \alpha$ for all $i \in S$. Then
$$ \sum_{i,j \in S} \| \partial_{ij} f \|_2^2 = O\left(\sqrt{\alpha} \log^{3/2} \frac{1}{\alpha} \right).$$
\end{lemma}

\iffull
\begin{proof}
First, note that coordinates $i \in S$ such that $\| \partial_i f \|_T = 0$ do not contribute anything to the sum $\sum_{i,j \in S} \| \partial_{ij} f \|_2^2$. This is because if  $\| \partial_i f \|_T = 0$ then $\partial_i f(x) = 0$ for all $x \in \zo^n$ and hence also $\partial_{ij} f(x) = 0$ for every other coordinate $j \in S$. Therefore, we can assume that $\|\partial_i f\|_T > 0$ for all $i \in S$.

Let us partition the coordinates as follows. For each $\ell \geq 0$, let
$$ B_k = \{ i \in S: \| \partial_i f \|_T > 2^{-k} \alpha \}.$$
Note that by assumption, $B_0 = \emptyset$, the sets $B_k$ form a chain and each $i \in S$ belongs to $B_k$ for large enough $k$. By Corollary~\ref{cor:threshold-norm}, the sets $B_k$ are bounded in size:
 $$ |B_k| = O\left(\frac{1}{2^{-k} \alpha} \log \frac{1}{2^{-k} \alpha}\right)
  = O\left(\frac{2^k}{\alpha} \log \frac{2^k}{\alpha} \right).$$
We define $A_k = B_{k+1} \setminus B_k$; clearly, each coordinate belongs to exactly one set $A_k, k \geq 0$,
and we have $\| \partial_i f \|_T \leq 2^{-k} \alpha$ for each $i \in A_k$.

We estimate $\sum_{i,j \in S} \| \partial_{ij} f \|_2^2$ as follows. We can write
$$ |\partial_{ij} f(x)| = |\partial_i f(x_{j \leftarrow 1}) - \partial_i f(x_{j \leftarrow 0})|
 \leq |\partial_i f(x_{j \leftarrow 1})| + |\partial_i f(x_{j \leftarrow 0})|. $$
Therefore,
$$ \| \partial_{ij} f \|_2^2 = \E_{x}[ |\partial_{ij} f(x)|^2 ]
 \leq \E_{x} [ (|\partial_i f(x_{j \leftarrow 1})| + |\partial_i f(x_{j \leftarrow 0})|)  \cdot |\partial_{ij} f(x)| ].$$
Assuming that $i \in A_k$, we know that $\| \partial_i f\|_T \leq 2^{-k} \alpha$, and hence $\Pr_x[|\partial_i f(x)| \geq 2^{1-k} \alpha] \leq 2^{-3k} \alpha^3$. Therefore, we also have
$\Pr_x[|\partial_i f(x_{i \leftarrow 1})| + |\partial_i f(x_{i \leftarrow 0})| \geq 2^{2-k} \alpha] \leq 2^{1-3k} \alpha^3$. Hence for $i \in A_k$ we can estimate
$$ \| \partial_{ij} f \|_2^2  \leq \E_{x} [ (|\partial_i f(x_{j \leftarrow 1})| + |\partial_i f(x_{j \leftarrow 0})|)  \cdot |\partial_{ij} f(x)| ]
\leq 2^{2-k} \alpha \cdot \E_{x} [|\partial_{ij} f(x)|] + 2^{2-3k} \alpha^3 $$
using the trivial bound that $|\partial_{ij} f(x)| \leq 2$ in the case where $|\partial_i f(x_{i \leftarrow 1})| + |\partial_i f(x_{i \leftarrow 0})|$ is large.

Overall, we obtain
\begin{eqnarray*}
 \sum_{i,j \in S} \| \partial_{ij} f \|_2^2
 & \leq & 2 \sum_{0 \leq \ell \leq k} \sum_{i \in A_k} \sum_{j \in A_\ell} \| \partial_{ij} f \|_2^2 \\
 & \leq & 2 \sum_{0 \leq \ell \leq k} \sum_{i \in A_k} \sum_{j \in A_\ell} \left( 2^{2-k} \alpha \cdot \E_{x} [|\partial_{ij} f(x)|] + 2^{2-3k} \alpha^3 \right) \\
 &  = & \sum_{k \geq 0} 2^{3-k} \alpha \sum_{i \in A_k}  \sum_{\ell=0}^{k} \sum_{j \in A_\ell} \|\partial_{ij} f \|_1
  + \sum_{0 \leq \ell \leq k} 2^{3-3k} \alpha^3 |A_k| \cdot |A_\ell| .
\end{eqnarray*}
Here we use the bounds $|A_k| \leq |B_{k+1}| = O(\frac{2^{k}}{\alpha} \log \frac{2^{k}}{\alpha})$ to estimate the second term. We get (up to constant factors) $\sum_{0 \leq \ell \leq k} 2^{-3k} \alpha \cdot 2^{k+\ell} (k + \log \frac{1}{\alpha}) (\ell + \log \frac{1}{\alpha}) \leq \sum_{k \geq 0} k^3 2^{-k} \alpha \log^2 \frac{1}{\alpha} = O(\alpha \log^2 \frac{1}{\alpha})$.
Hence, we get
\begin{equation}
\label{eq:partial-bound}
 \sum_{i,j \in S} \| \partial_{ij} f \|_2^2
 \leq \sum_{k \geq 0} 2^{3-k} \alpha \sum_{i \in A_k}  \sum_{\ell=0}^{k} \sum_{j \in A_\ell} \|\partial_{ij} f \|_1
  + O\left(\alpha \log^2 \frac{1}{\alpha} \right).
\end{equation}
Here we use Lemma~\ref{lem:square-root-bound} to estimate $\sum_{\ell=0}^{k} \sum_{j \in A_\ell} \| \partial_{ij} f \|_1$.
Recall that the $A_\ell$'s are disjoint and $\bigcup_{\ell=0}^{k} A_\ell = B_{k+1}$. By Lemma~\ref{lem:square-root-bound},
$$ \sum_{\ell=0}^{k} \sum_{j \in A_\ell} \| \partial_{ij} f \|_1
 = \sum_{j \in B_{k+1}} \| \partial_{ij} f \|_1 \leq 2 \sqrt{|B_{k+1}|} \| \partial_i f \|_2
 = O\left( \frac{2^{k/2}}{\sqrt{\alpha}} \log^{1/2} \frac{2^k}{\alpha} \right) \cdot \| \partial_i f \|_2. $$
Assuming $i \in A_k$, Lemma~\ref{lem:2vsT} says $\| \partial_i f \|_2 \leq \sqrt{2} \| \partial_i f \|_T \leq \sqrt{2} \frac{\alpha}{2^k}$.
Also, $|A_k| = O(\frac{2^k}{\alpha} \log \frac{2^k}{\alpha})$, so we get
$$ \sum_{i \in A_k}  \sum_{\ell=0}^{k} \sum_{j \in A_\ell} \|\partial_{ij} f \|_1
 = O\left( \frac{2^{k/2}}{\sqrt{\alpha}} \log^{1/2} \frac{2^k}{\alpha} \right) \cdot \sum_{i \in A_k} \| \partial_i f \|_2
 = O\left( \frac{2^{k/2}}{\sqrt{\alpha}} \log^{3/2} \frac{2^k}{\alpha} \right).$$
Continuing the computation from equation (\ref{eq:partial-bound}), we have
\begin{eqnarray*} \sum_{i,j \in S} \| \partial_{ij} f \|_2^2
& \leq & \sum_{k \geq 0} 2^{3-k} \alpha \sum_{i \in A_k}  \sum_{\ell=0}^{k} \sum_{j \in A_\ell} \|\partial_{ij} f \|_1
  + O\left(\alpha \log^2 \frac{1}{\alpha} \right) \\
 & = & O\left(\sum_{k \geq 0} \frac{\sqrt{\alpha}}{2^{k/2}} \log^{3/2} \frac{2^k}{\alpha} \right)
  + O\left(\alpha \log^2 \frac{1}{\alpha} \right) \\
 & = & O\left(\sqrt{\alpha} \log^{3/2} \frac{1}{\alpha} \right).
\end{eqnarray*}
\end{proof}

\fi

\iffull
\begin{theorem}
\label{thm:submod-lowdegree}
For any submodular function $f:\zo^n \rightarrow [0,1]$, there is a polynomial of degree $O(\frac{1}{\epsilon^{4/5}} \log \frac{1}{\epsilon})$ such that
$ \| f - p \|_2 \leq \epsilon.$
\end{theorem}

\begin{proof}
Let $\alpha = \epsilon^{4/5}$. Let $L$ be the set of variables $i \in [n]$ such that $\| \partial_i f \|_T > \alpha$.
By Corollary~\ref{cor:threshold-norm}, the number of such variables is $|L| = O(\frac{1}{\alpha} \log \frac{1}{\alpha}) = O(\frac{1}{\epsilon^{4/5}} \log \frac{1}{\epsilon})$. By Lemma~\ref{lem:submod-Lipschitz-bound} (for $S = [n] \setminus L$), we have $\sum_{i,j \notin L} \| \partial_{ij} f \|_2^2 = O(\sqrt{\alpha} \log^{3/2} \frac{1}{\alpha}) = O(\epsilon^{2/5} \log^{3/2} \frac{1}{\epsilon})$. On the other hand (recalling that $\|\partial_{ii} f\|_2 = 0$ and $\|\partial_{ij} f \|_2^2 = 16\sum_{S \supseteq \{i,j\}} \hat{f}^2(S)$ for $i \neq j$),
$$ \sum_{i,j \notin L} \| \partial_{ij} f \|_2^2
 =16 \sum_{S:|S \setminus L| \geq 2} |S \setminus L| (|S \setminus L|-1) \hat{f}^2(S)
  \geq 16 \sum_{S: |S \setminus L|>k} k^2 \hat{f}^2(S). $$
We set $k = \frac{1}{\epsilon^{4/5}} \log \frac{1}{\epsilon}$ and obtain
$$ \sum_{S: |S \setminus L|>k} \hat{f}^2(S) \leq \frac{1}{16k^2} \sum_{i,j \notin L} \| \partial_{ij} f \|_2^2
 = \frac{\epsilon^{8/5}}{\log^2 \frac{1}{\epsilon}} \cdot O\left( \epsilon^{2/5} \log^{3/2} \frac{1}{\epsilon} \right)
 = O\left(\epsilon^2 \log^{-1/2} \frac{1}{\epsilon} \right).$$
For $\epsilon>0$ sufficiently small, this is less than $\epsilon^2$. Therefore, the polynomial
$$ p(x) = \sum_{S: |S \setminus L| \leq k} \hat{f}^2(S) \chi_S(x) $$
satisfies
$$ \|f-p\|_2^2 = \sum_{S: |S \setminus L|>k} \hat{f}^2(S) < \epsilon^2 $$
and has degree $|L|+k = O(\frac{1}{\epsilon^{4/5}} \log \frac{1}{\epsilon})$.
\end{proof}

\fi 
\iffull
\section{Applications}
\label{sec:applications}
\subsection{Approximation of XOS functions by juntas}
We use several notions of {\em influence} of a variable on a real-valued function which are based on the standard notion of influence for Boolean functions (\eg \cite{Ben-OrLinial:85,KahnKL:88}).
\begin{definition}[Influences]
For a real-valued $f:\zo^n \rightarrow \RR$, $i \in [n]$, and $\kappa \geq 0$ we define the {\em $\ell_\kappa^\kappa$-influence} of variable $i$ as $\infl^\kappa_i(f) = \|\fr{2}\partial_i f\|_\kappa^\kappa = \E[|\fr{2}\partial_i f|^\kappa]$. We define $\infl^\kappa(f) = \sum_{i\in[n]} \infl^\kappa_i(f)$.
\end{definition}
The most commonly used notion of influence for real-valued functions is the $\ell_2^2$-influence which satisfies
$$\infl^2_i(f) = \left\|\fr{2}\partial_i f\right\|_2^2 = \sum_{S \ni i}\hat{f}^2(S)\ .$$
From here, the total $\ell_2^2$-influence is equal to $\infl^2(f) = \sum_S |S| \hat{f}^2(S)$.
We use the following generalization of Friedgut's theorem \cite{Friedgut:98} from \cite{FeldmanVondrak:13arxiv}.
\begin{theorem}[\cite{FeldmanVondrak:13arxiv}]
\label{th:rvjunta}
Let $f:\zo^n\rightarrow \RR$ be any function, $\eps \in (0,1)$ and $\kappa \in (1,2)$. For $d$ such that $\sum_{|S|>d} \hat{f}(S)^2 \leq \eps^2/2$, let $$I = \{i\in[n] \cond \infl^{\kappa}_i(f) \geq \alpha\} \mbox{ for}$$ $$\alpha = \left((\kappa-1)^{d-1} \cdot \eps^2/ (2 \cdot\infl^{\kappa}(f)) \right)^{\kappa/(2-\kappa)}\ .$$
Then for the set $\I_d =\{S \subseteq I \cond |S| \leq d\}$ we have $\sum_{S\not\in \I_d} \hat{f}(S)^2 \leq \eps^2$.
\end{theorem}
Finally, to apply this generalization we need a bound on the total influence of any XOS function:
\begin{lemma}
\label{lem:xos}
Let $f:\zo^n\rightarrow \RR_+$ be an XOS function. Then $\infl^1(f) \leq \|f\|_1$.
In particular, for an XOS function $f:\zo^n\rightarrow [0,1]$, for all $\kappa \geq 1$,
$\infl^\kappa(f) \leq \infl^1(f) \leq 1$.
\end{lemma}
Combining these results with the degree bounds from Corollary \ref{cor:XOS-degree-upper} gives the following bound:
\begin{corollary}
\label{cor:XOS-junta-bound}
Let $f:\zo^n \rightarrow [0,1]$ be an XOS function and $\eps>0$. There exists a $2^{O(1/\eps)}$-junta $p$ of Fourier degree $O(1/\eps)$, such that $\|f - p\|_2 \leq \epsilon$. In particular, the spectral $\ell_1$-norm of $p$ is $ \|\hat{p}\|_1 = \sum_{S \subseteq [n]} |\hat{p}(S)| = 2^{O(1/\eps^2)}$.
\end{corollary}
\begin{proof}
By Corollary \ref{cor:XOS-degree-upper} we can use $d = O(1/\eps)$ in Theorem \ref{th:rvjunta} and we choose $\kappa = 4/3$.
Let $\alpha = \left((1/3)^{d-1} \cdot \eps^2/ (2 \cdot\infl^{4/3}(f)) \right)^{2}$ be the lower bound on the influence of variables in the junta given in Theorem \ref{th:rvjunta}.
By Lemma~\ref{lem:xos}, $\infl^{4/3}(f) \leq 1$.
Note that $g = \sum_{S \in \I_d} \hat{f}(S) \chi_S$ is a function of Fourier degree $d$ that depends only on variables in $I$. Further, $\|f-g\|_2^2 \leq \eps^2$ and the set $I$ has size at most
\equn{ |I| \leq \infl^{4/3}(f)/\alpha \leq 3^{2(d-1)} \cdot (2/\eps^2)^2  = 2^{O(1/\eps)}.}
\end{proof}

\subsection{Applications to Learning}
\label{sec:applications-learning}
\subsubsection{Preliminaries: Models of Learning}
We consider two models of learning based on the PAC model \cite{Valiant:84} which assumes that the learner has access to random examples of an unknown function from a known class of functions. Here we only consider learning over the uniform distribution over $\zo^n$ and hence simplify the definitions for this setting.
\begin{definition}[PAC learning with $\ell_2$-error]
Let $\F$ be a class of real-valued functions on $\zo^n$. An algorithm $\A$ PAC learns $\F$ with $\ell_2$ error over $\U$, if given $\epsilon > 0$, for every target function $f \in \F$, given access to random  independent samples from $\U$ labeled by $f$, with probability at least $2/3$,  $\A$ returns a hypothesis $h$ such that $\|f - h \|_2 \leq  \epsilon$. 
\end{definition}

\begin{definition}[Agnostic learning with $\ell_2$-error]
Let $\F$ be a class of real-valued functions on $\zo^n$. For any distribution $\cal P$ over $\zo^n \times [0,1]$, let $\mbox{opt}(\cal P,\F)$ be defined as: $$\mbox{opt}({\cal P},\F) =  \inf_{f \in \F} \sqrt{\E_{(x,\ell) \sim {\cal P}} [ (\ell - f(x) )^2]} .$$ An algorithm $\A$, is said to agnostically learn $\F$ with $\ell_2$ {\em excess error} over $\U$ if for every $\epsilon> 0$ and any distribution $\cal P$ on $\zo^n \times [0,1]$ such that the marginal of $\cal P$ on $\zo^n$ is $\U$, given access to random independent examples drawn from $P$, with probability at least $\frac{2}{3}$, $\A$ outputs a hypothesis $h$ such that $$\sqrt{\E_{(x,\ell) \sim \cal P} [ (h(x)- \ell)^2 ]} \leq \mbox{opt}(\cal P, \F) + \epsilon.$$
\end{definition}
We remark that one can also define optimality with respect to labels from a different range. For simplicity we use the $[0,1]$ range since that is also the range of the functions we consider.

For both PAC and agnostic learning we will rely on the fact that polynomials of degree $d$ over $n$ variables can be learned agnostically in time polynomial in $(e \cdot n/d)^d$. For the uniform distribution this follows from the agnostic properties of the low-degree algorithm by Linial \etal \cite{LinialMN:93} observed by Kearns \etal \cite{KearnsSS:94}.
\begin{theorem}
\label{thm:low-degree-regression}
Let $\cH_d$ be a class of all degree $d$ polynomials over $n$ variables of $\ell_2$-norm at most 1.
Then $\cH_d$ can be learned agnostically over $\U$ with excess $\ell_2$ error of $\eps$ in time polynomial in $t$ and $1/\eps$, where $t = \sum_{i=0}^d {n \choose i} = O((e \cdot n/d)^d)$.
\end{theorem}
We remark that this result also holds over arbitrary distributions and follows from the standard uniform convergence bounds for linear models with squared loss (\eg \cite{KakadeST:08}).
\subsubsection{PAC and Agnostic Learning of Submodular and XOS Functions}
The algorithms for PAC learning of submodular and XOS functions in \cite{FeldmanVondrak:13arxiv} are based on two steps:
\begin{enumerate}
\item Identify a set of influential variables $J$ such that there exists a submodular (or XOS accordingly) function $h$ that depends only on variables in $J$ and is close to $f$.
\item Use regression over all parity functions of degree at most $d$ on variables in $J$ to find the polynomial that best fits sampled examples.
\end{enumerate}
For XOS functions the first step involves simply choosing variables with large enough Fourier coefficients of degree 1. The analysis of both of these steps in \cite{FeldmanVondrak:13arxiv} is in $\ell_2$ norm and therefore we can directly plug in our new bounds to obtain Theorem \ref{thm:XOS-learning-intro}.
\eat{
the following bound for learning XOS functions.
\begin{theorem}[Thm.~\ref{thm:XOS-learning-intro} restated]
\label{thm:learn-lowinfluence}
There exists an algorithm $\A$ that given $\eps > 0$ and access to random uniform examples of any XOS function $f:\zon \rightarrow [0,1]$, with probability at least $2/3$, outputs a function $h$, such that $\|f-h\|_2 \leq \epsilon$. $\A$ runs in time $\tilde{O}(n) \cdot r^{O(1/\eps)}$ and uses $\log n \cdot r^{O(1/\eps)}$ examples, where $r = \min\{n,2^{1/\eps}\}$.
\end{theorem}
}

In the case of submodular functions in \cite{FeldmanVondrak:13arxiv} the algorithm that finds the set of influential variables only ensures that there exists a function that depends on variables in $J$ and is close in $\ell_1$ distance to $f$. We therefore provide an analogous result for $\ell_2$. As in \cite{FeldmanVondrak:13arxiv} our algorithm selects all variables that have a large degree-1 or 2 Fourier coefficient (with different values of thresholds). The set of variables it returns is larger but analysis is substantially simpler than that in \cite{FeldmanVondrak:13arxiv}.

Before proceeding we will need a few simple definitions. For a real-valued $f$ over $\zo^n$ and $\eps \in [0,1]$ let $s_f(\eps)$ denote the smallest $s$ such that there exists an $s$-junta $g$ for which $\|f-g\|_2 \leq \eps$. For a set of indices $J \subseteq [n]$ we say that a function is a $J$-junta if it depends only on variables in $J$. For a function $f$ and a set of indices $I$, we define the {\em projection} of $f$ to $I$ to be the function over $\zo^n$ whose value depends only on the variables in $I$ and its value at $x_I$ is the expectation of $f$ over all the possible values of variables outside of $I$, namely $f_I(x) = \E_{y\sim \U}[f(x_I,y_{\bar{I}})]$. Observe that an equivalent representation of $f_I$ is as follows: $$f_I(x) = \sum_{S\subseteq I} \hat{f}(S) \chi_S(x).$$ We will also need the following bound on the number of variables with large degree-1 or degree-2 Fourier coefficient from \cite{FeldmanVondrak:13arxiv} and the property of degree-2 Fourier coefficient of submodular functions from \cite{FeldmanKV:13}.
\begin{lemma}[\cite{FeldmanVondrak:13arxiv}]
\label{lem:important-var-bound}
Let $f:\zo^n \rightarrow [0,1]$ be a submodular function and $\alpha,\beta > 0$. Let
$$I = \left\{ i  \left|\ |\hat{f}(\{i\})| \geq \alpha \right.\right\} \bigcup \left\{ i\ \left|\ \exists j, |\hat{f}(\{i,j\})| \geq \beta \right.\right\}\ .$$ Then $|I| \leq \frac{2}{\min\{\alpha,\beta\}}$.
\end{lemma}
\begin{lemma}[\cite{FeldmanKV:13}]
\label{lem:upp-bound-sum}
Let $f:\zon \rightarrow [0,1]$ be a submodular function and $i,j \in [n]$, $i\neq j$.
\equn{|\hat{f}(\{i,j\})| \geq \frac12 \sum_{S \ni i,j} (\hat{f}(S))^2  .}
\end{lemma}
We now state the guarantees of our algorithm for finding relevant variables.
\begin{theorem}
\label{thm:find-junta-examples}
Let $f:\zo^n \rightarrow [0,1]$ be a submodular function. There exists an algorithm, that given any $\eps > 0$ and access to random and uniform examples of $f$, with probability at least $5/6$, finds a set of variables $I$ of size at most $32 \cdot (s_f(\eps/2))^2/\eps^2$ such that there exists a submodular $I$-junta $h$ satisfying $\|f -h\|_2 \leq \eps$. The algorithm runs in time $O(n^2 \log (n) \cdot (s_f(\eps/2))^4/\eps^4)$ and uses $O(\log (n) \cdot (s_f(\eps/2))^4/\eps^4)$ examples.
\end{theorem}
\begin{proof}
Denote $s = s_f(\eps/2)$ and let $J$ be the set of variables such that there exists a $J$-junta $g$ such that $\|f-g\|_2 \leq \eps/2$. We can assume without loss of generality that $g = f_J$ since $f_J$ is a submodular $J$-junta and it is the $J$-junta closest to $f$ in $\ell_2$ distance.
Let $$I' = \left\{ i  \left|\  |\hat{f}(\{i\})| \geq \frac{\eps}{4 \cdot \sqrt{s}} \right.\right\} \bigcup \left\{ i  \left|\  \exists j, |\hat{f}(\{i,j\})| \geq \frac{\eps^2}{8 \cdot s^2} \right.\right\}\ .$$
We claim that $\|f_J - f_{I'\cap J} \|_2 \leq \eps/2$. Note that by triangle inequality this would imply that $\|f - f_{I'\cap J} \|_2 \leq \eps$ meaning that it would suffice to find the variables in $I'$.

Using Lemma \ref{lem:upp-bound-sum} and the definition of $I'$, we prove the claim as follows:
\alequn{
\|f_J - f_{I'\cap J} \|_2^2 & = \sum_{S \subseteq J,\ S \not\subseteq I'} \hat{f}(S)^2
\\ & = \sum_{i \in J\setminus I'} \hat{f}(\{i\})^2 + \sum_{S \subseteq J,\ S \not\subseteq I',\ |S| \geq 2} \hat{f}(S)^2
\\ & \leq  |J\setminus I'| \cdot \frac{\eps^2}{16 \cdot s} + \sum_{i,j \in J,\ \{i,j\} \not\subseteq I',\ i>j}  \sum_{S \subseteq J,\ S \ni i,j } \hat{f}(S)^2
\\ & \leq \frac{\eps^2}{16} + \sum_{i,j \in J,\ \{i,j\} \not\subseteq I',\ i>j} 2 \cdot  |\hat{f}(\{i,j\})|
\\ & \leq \frac{\eps^2}{16} + \frac{|J|^2}{2} \cdot 2 \cdot \frac{\eps^2}{8 \cdot s^2} \leq \frac{\eps^2}{4}.
}

All we need now is to find a small set of indices $I \supseteq I'$. We simply estimate degree-1 and 2 Fourier coefficients of $f$ to accuracy $\eps^2/(32 \cdot s^2)$ with confidence at least $5/6$ using random examples. Let $\tilde{f}(S)$ for $S\subseteq [n]$ of size 1 or 2 denote the obtained estimates. We define
$$I = \left\{ i \lcond |\tilde{f}(\{i\})| \geq \frac{3\eps}{16 \cdot \sqrt{s}} \right.\right\} \bigcup \left\{ i \lcond \exists j, |\tilde{f}(\{i,j\})| \geq \frac{3\eps^2}{32 \cdot s^2} \right.\right\}\ .$$
If estimates are within the desired accuracy, then clearly, $I \supseteq I'$. At the same time $I \subseteq I''$, where
$$I'' = \left\{ i \lcond |\hat{f}(\{i\})| \geq \frac{\eps}{8\cdot \sqrt{s}} \right.\right\} \bigcup \left\{ i \lcond \exists j, |\hat{f}(\{i,j\})| \geq \frac{\eps^2}{16 \cdot s^2} \right.\right\}\ .$$ By Lem.~\ref{lem:important-var-bound}, $|I''| \leq 32 \cdot s^2/\eps^2$.

Finally, to bound the running time we observe that, by the standard application of Chernoff bound with the union bound, $O(\log (n) \cdot  s^4/\eps^4)$ random examples are sufficient to obtain the desired estimates with confidence of $5/6$. The estimation of the coefficients can be done in $O(n^2 \log(n) \cdot s^4/\eps^4)$ time.
\end{proof}

We can now use the result from \cite{FeldmanVondrak:13arxiv} that for every submodular function $f:\zon \rightarrow [0,1]$, $s_f(\eps/2) = O(\log(1/\eps)/\eps^2)$ to obtain the following corollary.
\begin{corollary}
\label{cor:find-junta-examples}
Let $f:\zo^n \rightarrow [0,1]$ be a submodular function. There exists an algorithm, that given any $\eps > 0$ and access to random and uniform examples of $f$, with probability at least $5/6$, finds a set of variables $I$ of size $\tilde{O}(1/\eps^{6})$ such that there exists a submodular $I$-junta $h$ satisfying $\|f -h\|_2 \leq \eps$. The algorithm runs in time $\tilde{O}(n^2 /\eps^{12})$ and uses $\tilde{O}(\log (n) / \eps^{12})$ examples.
\end{corollary}

We use Corollary \ref{cor:find-junta-examples} with the least squares regression over polynomials of degree $O(\log(1/\eps)/\eps^{4/5})$ on the influential variables to obtain the learning algorithm claimed in Theorem \ref{thm:submod-learning-intro}.

Finally, for completeness we also state the corollaries for agnostic learning of XOS and submodular functions:
\begin{theorem}
Let $\C_s$ be the class of all submodular functions with range in $[0,1]$. There exists an algorithm that learns $\C_s$ agnostically with excess $\ell_2$-error $\eps$ and runs in time $n^{O(\log(1/\eps)/\eps^{4/5})}$.
\end{theorem}
\begin{theorem}
Let $\C_x$ be the class of all XOS functions with range in $[0,1]$. There exists an algorithm that learns $\C_x$ agnostically with excess $\ell_2$-error $\eps$ and runs in time $n^{O(1/\eps)}$.
\end{theorem}

\newcommand{\maj}{\mathsf{maj}}
\newcommand{\hs}{\mathsf{hs}}
\newcommand{\NS}{\mathsf{NS}}
\section{Lower Bounds}
\label{sec:lower-bounds}
In this section we prove tight lower bounds on low-degree spectral concentration and learning of monotone submodular, XOS and self-bounding functions. 

\subsection{Monotone Submodular Functions}
We start by showing that for any $\eps > 0$ there exists an explicit monotone submodular function over $\Theta(\eps^{-4/5})$ variables that requires degree $\Theta(\eps^{-4/5})$ to $\ell_2$-approximate within $\eps$. The ``hockey-stick" function of $k$ (out of $n$) variables is defined as follows: $\hs_k(x)  = \min\left\{1, 2 \cdot w_k(x)/k\right\}$, where $w_k(x) = \sum_{i=1}^{k} x_i$ is the Hamming weight of the first $k$ bits of $x$. In \cite{FeldmanKV:13} it was shown that this function has a Fourier coefficient of degree $k$ whose value is at least $\Omega(k^{-3/2})$. This immediately implies a lower bound of $\Omega(\eps^{-2/3})$ on $\aedeg(\hs_k)$ for $k=\Theta(\eps^{-2/3})$. We now give a more careful analysis of the low-degree spectral concentration of $\hs_k$ that leads to the nearly tight lower bound.

The hockey-stick function is closely related to the well-studied Boolean majority function for which tight spectral concentration bounds are known \cite{ODonnell14:book}. Specifically, it is easy to see that for every $i$,
\equ{\partial_i \hs_k(x) = 2(1-\maj_k(x))/k, \label{eq:hs2maj}} where $\maj_k(x)  = 1$ if $w_k(x) \geq k/2$ and $0$ otherwise. This correspondence allows us to easily obtain a lower bound on the low-degree spectral concentration of $\hs_k(x)$ .

\begin{lemma}
\label{lem:submod-conc-lower}
For any $k \leq n$ and $d \leq k/2$, $W^{> d}(\hs_k) = \Omega(d^{-3/2}/k)$. In particular, for some constant $c_1$, $k = c_1 \eps^{-4/5}$ and $d = \lfloor k/2 \rfloor$ gives $W^{> d}(\hs_k) \geq \eps^2$.
\end{lemma}
\begin{proof}
We first observe that by the properties of partial derivatives given in Sec.~\ref{sec:prelims} and eq.\eqref{eq:hs2maj}, for every $S \subseteq [k]$ such that $|S| \geq 2$ and $i\in S$, $$\widehat{\hs_k}(S) = -\widehat{\partial_i \hs_k}(S\setminus i)/2 = \frac{\widehat{\maj_k}(S\setminus i)}{k}.$$
For $2 \leq j \leq k$,
\alequn{W^j(\hs_k) &= \sum_{S\subseteq[k],\ |S| =j}\widehat{\hs_k}(S)^2 = \sum_{S\subseteq[k],\ |S| =j,\ i \in S}\frac{\widehat{\maj_k}(S\setminus i)^2}{k^2} \\ &= \frac{k-j+1}{j}\sum_{S\subseteq[k],\ |S|=j-1}\frac{\widehat{\maj_k}(S)^2}{k^2}
= \frac{k-j+1}{k^2j}\cdot W^{j-1}(\maj_k)}
We now use the estimate $W^{j-1}(\maj_k) \geq c (j-1)^{-3/2}$ for some constant $c > 0$ \cite{ODonnell14:book}.
This estimate implies that
\alequn{W^{> d}(\hs_k) &\geq \sum_{2k/3+1 \geq j > d} W^{j}(\hs_k) = \sum_{2k/3+1 \geq j > d} \frac{k-j+1}{k^2j}\cdot W^{j-1}(\maj_k)\\& \geq c \sum_{2k/3 \geq j \geq d}\frac{k-j}{k^2} j^{-5/2} \geq  \frac{c}{3k} \sum_{2k/3 \geq j \geq d} j^{-5/2} \geq  \frac{c}{3k} \int_d^{2k/3} t^{-5/2} dt \\ &\geq \frac{c}{3k} \cdot \frac{2}{3} \left( d^{-3/2} - (2k/3)^{-3/2}\right) \geq \frac{2c}{27}\cdot \frac{d^{-3/2}}{k}, } where in the last inequality we used the condition that $d \leq k/2$ and hence $d^{-3/2} - (2k/3)^{-3/2} \geq d^{-3/2}/3$.
\end{proof}

We now show that any algorithm that PAC learns monotone submodular functions with $\ell_2$ error of $\epsilon$ must use $2^{\Omega(\epsilon^{-4/5})}$ examples. This result is based on a reduction from  learning the class all Boolean functions on $k$ variables with error $1/4$ to the problem of learning submodular functions on $2t = k + \lceil \log{k} \rceil + O(1)$ variables with $\ell_2$ error of $\Theta(\frac{1}{t^{5/4}})$. Any algorithm that learns the class of all Boolean functions on $k$ variables to accuracy $1/4$ requires at least $2^{\Omega(k)}$ bits of information about the target function and, in particular, at least that many random examples or other Boolean-valued queries are necessary.
The reduction is identical to the reduction in \cite{FeldmanKV:13} which proved an analogous result for learning with $\ell_1$ error of $\Theta(\frac{1}{t^{3/2}})$. Therefore the analysis of the reduction follows closely that from \cite{FeldmanKV:13}.



\begin{lemma}
For $k > 0$, let $t > 0$ be the smallest such that ${{2t} \choose t} \geq 2^k$ (and thus $4\cdot 2^k > {{2t} \choose t} \geq 2^k$). For every Boolean function $h:\zo^k \rightarrow \zo$ there exists a monotone submodular function $f:\zo^{2t} \rightarrow [0,1]$ such that:
\begin{enumerate}
\item $f$ can be computed at any point $x \in \zo^{2t}$ in at most a single query to $h$ and in time $O(k)$; given a single random and uniform example of $h$, a random and uniform example of $f$ can be produced in time $O(k)$.
\item Let $\alpha = \frac{2^k\cdot \sqrt{t}}{2^{2t}} = \Theta(1)$. For any $\beta > 0$, given a function $\tilde{f}:\zo^{2t} \rightarrow \R$ such that $\|f - \tilde{f}\|_2 \leq \frac{\sqrt{\alpha\beta}}{4\cdot t^{5/4}}$, one can obtain a Boolean function $\tilde{h}:\zo^k \rightarrow \zo$ such that $\pr_{\U} [ \tilde{h}(x) \neq h(x)] \leq \beta$ and $\tilde{h}$ can be computed at any point $x \in \zo^k$, with a single query to $\tilde{f}$ in time $O(k)$.
\end{enumerate}\label{lem:embed}
\end{lemma}
\begin{proof}
We construct $f$ by embedding $h$ into the middle layer of $\hs=\hs_{2t}$ while preserving the monotonicity and submodularity. The embedding modifies the values of $\hs$ by at most $\frac{1}{2t}$.

Let $s = {{2t} \choose t}$. Let $M_{t} = \{x  \in \zo^{2t}\mid w_{2t}(x) = t\}$ be the middle layer of $\zo^{2t}$ and let $m: \zo^k \rightarrow M_{t}$ be an injective map such that both $m$ and $m^{-1}$ (whenever it exists) can be computed in time $O(k)$ at any given point (for example using lexicographic ordering on both sets). We now define $f$ as:

\[
f(x)=\left\{\begin{array}{cl}
	\hs(x) & x \not\in M_t  \\
          1-\frac{1-h(y)}{2t} & x\in M_t \text{ and } \exists y \in \zo^k,\ m(y) = x \\
          1 & \text{ otherwise }\\
	   \end{array}\right.
\]

Notice that given any $x \in \zo^{2t}$, the value of $f(x)$ can be computed using a single query to $h$ and it is easy to see that given a single random and uniform example of $h$ we can output a random and uniform example of $f$ in time $O(k)$.

Given a function $\tilde{f}: \zo^{2t} \rightarrow \R$, define $\tilde{h}:\zo^k \rightarrow \zo$ so that $\tilde h(y) = 1$ if $\tilde{f}(m(y))  \geq (1-(1/4t))$ and $\tilde h(y) = 0$ otherwise. By definition,
$$ \tilde{h}(y) \neq h(y) \Rightarrow |\tilde{f}(m(y))-f(m(y))| \geq \frac{1}{4t}. $$
 Using that $\pr_{x \sim \U_{2t}} [ \exists y,\ m(y) = x] = \frac{\alpha}{\sqrt{t}}$, we have:
\begin{align*}
 \pr_{\U_k}[\tilde h(y) \neq h(y)] & \leq \pr_{x \sim \U_{2t}}[ |\tilde{f}(x)-f(x)| > 1/4t \cond \exists y,\ m(y) = x ] \\ &\leq (4t)^2 \cdot \E_{x \sim \U_{2t}}[ |\tilde{f}(x)-f(x)|^2  \cond \exists y,\ m(y) = x ] \\
    & \leq (4t)^2 \cdot \frac{\E_{x \sim \U_{2t}}[ |\tilde{f}(x)-f(x)|^2 ]}{\frac{\alpha}{\sqrt{t}}} = \frac{16 \cdot t^{5/2}}{\alpha} \cdot \|\tilde{f}-f\|_2^2.
 \end{align*}

Using $\|\tilde{f}-f\|_2 \leq \frac{\sqrt{\alpha\beta}}{4\cdot t^{5/4}}$, we have: $ \E_{x \sim \U_k} [\tilde{h}(x) \neq h(x)] \leq \beta$.

Now, observe that $\hs$ is monotone and $f$ is obtained by modifying $\hs$ only on points in $M_{t}$ and by at most $\frac{1}{2t}$, which ensures that for any $x \leq y$ such that $w_{2t}(x) < w_{2t}(y)$, $f(x) \leq f(y)$. 
Finally, we show that $f$ is submodular for any Boolean function $h$. It will be convenient to switch notation and look at input $x$ as the indicator function of the set $S_x= \{ x_i \mid x_i = 1\}$. We will verify that for each $S \subseteq [n]$ and $i,j \notin S$,
\begin{equation}
f(S \cup \{i\}) - f(S) \geq f(S \cup \{i,j\})  - f(S \cup \{j\}).\label{submodularity}
 \end{equation}
Notice that $\hs$ is submodular, and $f = \hs$ on every $x$ such that $w_{2t}(x) \neq t$. Thus, we only need to check eq.\eqref{submodularity} for $S, i, j$ such that $|S| \in \{ t-2, t-1, t\}$.
We analyze these $3$ cases separately:
\begin{enumerate}
\item
 $\bm{|S| = t-1:}$ Notice that $f(S) = \hs(S) = 1-(1/t)$ and $f(S \cup \{i, j\}) = \hs( S \cup \{i, j\}) = 1$. Also observe that for any $h$, $f(S \cup \{i\})$ and $f(S \cup \{j\})$ are at least $(1-\frac{1}{2t})$. Thus, $f(S \cup \{i\}) + f(S \cup \{j\}) \geq 2 - \frac{1}{t} = f(S) + f(S \cup \{i,j\})$.
 \item
 $\bm{|S| = t-2:}$ In this case, $f(S) = (1-(2/t))$ and $f(S \cup \{i\}) = f(S \cup \{j\}) = (1-(1/t))$. In this case, the maximum value for any $h$, of $f(S \cup \{i,j\}) = 1$. Thus, $$f(S) +  f(S \cup \{i,j\})  \leq 2 - (2/t) = f(S \cup \{i\}) + f(S \cup \{j\}).$$
 \item
$\bm {|S| = t:}$ Here, $f(S \cup \{i\}) = f(S \cup \{j\}) =  f(S \cup \{i,j\}) = 1$. The maximum value of $f(S)$ for any $h$ is $1$. Thus, $$f(S) +  f(S \cup \{i,j\})  \leq 2 = f(S \cup \{i\}) + f(S \cup \{j\}).$$
 \end{enumerate}
This completes the proof that $f$ is submodular.
\end{proof}

By choosing $\beta =1/4$ in Lemma \ref{lem:embed} we obtain the following result:

\begin{theorem}
\label{thm:submod-PAC-hardness}
Any algorithm that PAC learns all monotone submodular functions with range $[0,1]$ to $\ell_2$ error of $\eps>0$ requires $2^{\Omega(\eps^{-4/5})}$ random examples of (or value queries to) the target function.
\end{theorem}
\eat{
\begin{proof}
We borrow notation from the statement of Lemma \ref{embed} here. Given an algorithm that PAC learns monotone submodular functions on $2t$ variables, we describe how one can obtain a learning algorithm for all Boolean function on $k$ variables with accuracy $1/4$.
Given an access to a Boolean function $f:\zo^k \rightarrow \zo$, we can translate it into an access to a submodular function $h$ on $2t$ variables with an overhead of at most $O(t) = O(k)$ time using Lemma \ref{embed}. Using the PAC learning algorithm, we can obtain a function $g:\zo^{2t} \rightarrow \R$ that approximates $h$ within an error of at most $\alpha \cdot \frac{1}{8t^{3/2}}$ and Lemma \ref{embed} shows how to obtain $\tilde{f}$ from $g$ with an overhead of at most $O(t) = O(k)$ time such that $\tilde{f}$ approximates $f$ within $\frac{1}{4}$.
Choose $k =  \lceil \epsilon^{-2/3} \rceil$ and $t$ as described in the statement of Lemma \ref{embed}.
Now, using any algorithm that learns monotone submodular functions to an accuracy of $\epsilon >0$ we obtain an algorithm that learns all Boolean functions on $k = \lceil \epsilon^{-2/3} \rceil$ variables to accuracy $1/4$.
\end{proof}
}

\subsection{XOS functions}
\label{sec:lower-xos}
The lower bounds for XOS functions are based on a simple mapping from monotone DNF (MDNF) formulas to XOS functions. We say that a function $h$ is $s$-term $t$-MDNF if $h(x)= \bigvee_{j \in [s]} T_j(x)$, where each $T_j \subseteq [k]$, $|T_j| \leq t$ and $T_j(x) = \bigwedge_{i \in T_j} x_i$.
\begin{lemma}
For every $s$-term $t$-MDNF $h:\zo^k \rightarrow \zo$, let $f:\zo^k \rightarrow [0,1]$ be given by $f(x) = 1-\frac{1 - h(x)}{t}$, if $x \neq \bf 0$ and $f(x)=0$ otherwise. Then $f$
is an XOS function of size $s+k$.
\label{lem:xosdnf}
\end{lemma}
\begin{proof}
Let $h(x)= \bigvee_{j \in [s]} T_j(x)$ be an $s$-term $t$-MDNF representation of $h$. Then it is easy to verify that $$f(x)= \max\left\{\max \frac{\sum_{i \in T_j} x_i}{|T_j|},\ \max_{i \in [k]}
\frac{t-1}{t} x_i
  \right\}.$$
\end{proof}

An immediate corollary of Lemma \ref{lem:xosdnf} is that for any $\beta > 0$, a function $g$ such that $\|f-g\|_2 \leq \sqrt{\beta}/(2t)$ gives a function $\tilde h$ such that $\pr_{\U} [ \tilde{h}(x) \neq h(x)] \leq \beta + 2^{-k}$.


To obtain our lower bounds, we rely on known results for MDNFs obtained by choosing random conjunctions of size $\Theta(\sqrt{k})$. Such MDNFs were first analyzed by Talagrand \cite{Talagrand:96}. For our spectral concentration lower bound we will use the fact that Talagrand's DNFs are noise sensitive \cite{mossel2002noise} together with a reverse connection between noise sensitivity and low-degree spectral concentration.

We first recall the definition and basic properties of the noise sensitivity.
\begin{definition}[Noise sensitivity]
For $\alpha \in [0,1], x \in \zon$, we define a distribution $N_\alpha(x)$ over $y \in \zon$ by letting $y_i = x_i$ with probability $1-\alpha$ and $y_i = 1-x_i$ with probability $\alpha$, independently for each $i$. For a Boolean function $h$, the noise sensitivity of $h$ with noise rate $\alpha$ is defined as
$$\NS_\alpha(h) = \pr_{x\sim \U,\ y \sim N_\alpha(x)} [h(x) \neq h(y)].$$
Noise sensitivity satisfies (\eg \cite{ODonnell14:book}):
\equ{\NS_\alpha(h) = \frac{1}{2} \sum_{i=0}^k  (1-(1-2\alpha)^i) \cdot W^i(h). \label{eq:ns-fourier}}
\end{definition}

The following theorem was proved in \cite{mossel2002noise}, following Talagrand's analysis \cite{Talagrand:96}.

\begin{theorem}[\cite{mossel2002noise}]
\label{thm:tal-is-noisy}
For every $k$, there exists a $\sqrt{k}$-MDNF $h$ such that $\NS_{1/\sqrt{k}}(h) = \Omega(1)$.
\end{theorem}

This result implies that such functions have a large Fourier mass above level $\Omega(\sqrt{k})$.

\begin{corollary}
\label{cor:spectralmdnf}
For every $k$, there exists a $\sqrt{k}$-MDNF $h$ such that for $d=\Omega{(\sqrt{k})}$, $W^{>d}(h) = \Omega(1)$.
\end{corollary}
\begin{proof}
Equation \eqref{eq:ns-fourier} implies that for every $d$,
\alequn{\NS_\alpha(h) &= \frac{1}{2} \sum_{i=0}^k  (1-(1-2\alpha)^k) \cdot W^k(h) \leq
\frac{1}{2} \sum_{i=0}^d  (1-(1-2\alpha)^d) \cdot W^i(h) + \frac{1}{2} W^{>d}(h)
\\ &\leq \frac{1}{2} \left( (1-(1-2\alpha)^d) \|h\|_2^2 + W^{>d}(h) \right) < \frac{1}{2} \left( 2\alpha d \cdot \|h\|_2^2 + W^{>d}(h) \right) \\ &= \alpha d \cdot \|h\|_2^2 + W^{>d}(h)/2 \leq \alpha d  + W^{>d}(h)/2 .}
By Theorem \ref{thm:tal-is-noisy}, there exists a $\sqrt{k}$-MDNF $h$ such that for some constant $c > 0$,
$\NS_{1/\sqrt{k}}(h) \geq c$. Let $d = c\sqrt{k}/2$ we obtain that $$W^{>d}(h) \geq 2\left(\NS_{1/\sqrt{k}}(h) - \frac{d}{\sqrt{k}}\right) \geq c.$$
\end{proof}

From here we obtain a lower bound on low-degree spectral concentration of XOS functions using Lemma \ref{lem:xosdnf}.
\begin{theorem}
For every $\eps>0$ there exists $k = \Theta(1/\eps^2)$ and an XOS function $f:\zo^k \rightarrow [0,1]$ such that $\aedeg(f)=\Omega(1/\epsilon)$.
\end{theorem}
\begin{proof}
For $k > 0$, let $h$ be the $\sqrt{k}$-MDNF $h$ such that for $d=\Omega{(\sqrt{k})}$, $W^{>d}(h) = \Omega(1)$. Let $f$ be the XOS function obtained from $h$ using Lemma \ref{lem:xosdnf}. Then, by the linearity of Fourier coefficients and the fact that $f$ differs from $1-\frac{1 - h(x)}{t}$ only on a single point, we obtain that
$$W^{>d}(f) \geq W^{>d}(h)/d^2 - 2^{-k} = \Omega(1/k).$$ This means that for some $k = \Theta(1/\eps^2)$ and $d=\Omega(1/\epsilon)$ we have $W^{>d}(f) \geq \eps^2$.
\end{proof}

Our lower bound for PAC learning of XOS functions is based on the following lower bound for learning MDNF by Blum \etal \cite{BlumBL:98}.

\begin{theorem}[\cite{BlumBL:98}]
\label{thm:bbl-hard}
For any sufficiently large $k$ and $q \geq k$, any algorithm that PAC learns $t$-MDNF for $t = \log(3qk)$ over the uniform distribution and uses at most $q$ random examples (or value queries) will have error of at least $1/2 - O(\log(qk)/\sqrt{k})$.
\end{theorem}

We note that Theorem \ref{thm:bbl-hard} implies a slightly weaker (by a logarithmic factor in the degree) version of Corollary \ref{cor:spectralmdnf} since low-degree spectral concentration implies learning (in fact, as shown in \cite{DachmanFTWW:15} this argument also implies a lower bound on $\ell_1$-approximation by polynomials).
We now prove a lower bound for PAC learning XOS functions which we state for the $\ell_1$ error (which implies the same lower bound for $\ell_2$ error).

\begin{theorem}
\label{thm:xos-pac-hardness}
Any algorithm that PAC learns all XOS functions from $\zon$ to $[0,1]$ with $\ell_1$ error of $\eps>0$ requires $2^{\Omega(1/\eps)}$ random examples of (or value queries to) the target function.
\end{theorem}
\begin{proof}
We reduce learning of $t$-MDNF over $k$ variables (for $t$ and $k$ to be chosen later) to learning of XOS using Lemma \ref{lem:xosdnf}, namely we replace each example $(x,f(x))$ with $\left(x, 1-\frac{1 - f(x)}{t}\right)$ and then replace the hypothesis $h(x)$ with $h'$ such that $h'(x)=1$ whenever $h(x) \geq 1-1/(2t)$. By Lemma \ref{lem:xosdnf}, any algorithm that achieves $\ell_1$ error of $\frac{1/4}{2t} - 2^{-k}$ gives a Boolean hypothesis for the MDNF problem with error of less than $1/4$.
 
By Theorem \ref{thm:bbl-hard}, there exists a constant $c > 0$ such that for $q= 2^{c \sqrt{k}}$ and $t= \log(3qk)$, the error of any PAC learning algorithm for $t$-MDNF that uses at most $q$ random examples (or value queries) is at least $1/4$. Note that  $$\frac{1/4}{2t} - 2^{-k}  = \frac{1}{8\log(3qk)} - 2^{-k} = \frac{1}{8(\log(3k)+c\sqrt{k})} - 2^{-k}, $$ and therefore there exists a constant $c_1 >0$ such that for every $\eps > 0$ and $k = c_1/\eps^2$, $\frac{1/4}{2t} - 2^{-k} \geq \eps$.
Applying the guarantees of Theorem \ref{thm:bbl-hard}, we get that the number of random examples (or value queries) used to learn with $\ell_1$ error of $\eps$ must be larger than $q= 2^{c  \sqrt{k}} = 2^{\Omega(1/\eps)}$.

\end{proof}

\subsection{Self-bounding functions}
We now show that upper bounds on low-degree spectral concentration that we proved for XOS and submodular functions cannot be extended to the whole class of self-bounding functions.
Our construction is based on the classical Hamming code which we briefly describe here for completeness. For an integer $r$ a Hamming code is a linear mapping (over $\gfield{2}$) $c:\zo^{2^r-r-1} \rightarrow \zo^{r}$ such that for any two distinct $v,w\in \zo^{2^r-r-1}$, the Hamming distance between $v \circ c(v)$ and $w \circ c(w)$ is at least 3, where we use ``$\circ$" to denote the concatenation of strings. We now show that for $k = 2^r-r-1$ a Hamming code gives a way to embed any Boolean function into a self-bounding function which we describe below.

\begin{lemma}
For an integer $r$, $k = 2^r-r-1$ and any Boolean function $h:\zo^k \rightarrow \zo$ let $f:\zo^{k+r} \rightarrow [0,1]$ be given by $f(x \circ z) = h(x)$, if $z = c(x)$ and $f(x \circ z) =1$, otherwise. Then $f$ is a self-bounding function.
\label{lem:sb-embed}
\end{lemma}
\begin{proof}
Let $x \circ z$ be a point in $\zo^{k+r}$. If $f(x \circ z) = 0$ then $f$ cannot be lower on any point that differs from $x \circ z$ in one coordinate, and therefore the self-bounding condition holds at $x \circ z$. If $f(x \circ z) = 1$ then there exists at most one point $y \in \zo^{k+r}$ that differs from $x \circ z$ in a single coordinate and $f(y) = 0$. This follows from the fact that, by definition of $f$, if $f(y) = 0$ then $y = x' \circ c(x')$ for some $x'\in \zo^k$. By the properties of $c$, any two points of this form are at Hamming distance at least 3 and therefore two distinct points cannot be at Hamming distance 1 to $x \circ z$. This means that $$\sum_{i \in [k+r]} |f(x \circ z) - f((x \circ z) \oplus e_i)|   \leq 1 = f(x \circ z).$$
\end{proof}

A spectral concentration bound can be obtained by analyzing the embedding of a $\zo$-parity function $h= \sum_{i \in S} x_i \mod 2$. To avoid the direct calculation which requires using additional properties of the Hamming code we will derive the lower-bound via lower bounds for learning below.

\begin{theorem}
Any algorithm that PAC learns all self-bounding functions from $\zon$ to $[0,1]$ with $\ell_2$ error of $\eps>0$ requires $2^{\Omega(1/\eps^2)}$ random examples of (or value queries to) the target function.
\end{theorem}
\begin{proof}
We reduce learning of all Boolean functions on $k=2^r-r-1$ (for $r$ to be chosen later) variables over the uniform distribution to learning of self-bounding functions using Lemma \ref{lem:sb-embed}. Namely, given a random and uniform example $(x,\ell)$ of some unknown Boolean target function $h$ we output a random example $(x \circ z, \ell')$ of the function $f$ that is equal to the embedding of $h$ given by Lemma \ref{lem:sb-embed}. This is done by choosing $z$ uniformly from $\zo^{r}$ and having $\ell'=\ell$ if $z = c(x)$ and $\ell'=1$ otherwise (a value query can be answered similarly using a single value query to $h$).
Given a hypothesis $\tilde f$ we define $\tilde h(x) = 1$ if $\tilde f(x \circ c(x)) \geq 1/2$ and $\tilde h(x) = 0$ otherwise.
Observe that,
\begin{align*}
 \pr_{\U_k}[\tilde h(y) \neq h(y)] & \leq \pr_{x \circ z \sim \U_{k+r}}[ |\tilde{f}(x \circ z)-f(x\circ z)| \geq 1/2 \cond c(x) = z ] \\ &\leq 4 \cdot \E_{x \circ z \sim \U_{k+r}}[ |\tilde{f}(x \circ z)-f(x \circ z)|^2   \cond c(x) = z  ] \\
    & \leq 4 \cdot \frac{\E_{x \circ z \sim \U_{k+r}}[ |\tilde{f}(x \circ z)-f(x \circ z)|^2 }{2^{-r}} = 2^{r+2} \cdot \|\tilde{f}-f\|_2^2.
 \end{align*}
We now let $r= \lfloor\log(1/\eps^2)\rfloor + 4$. This choice ensures that if $\tilde f$ has $\ell_2$ error of less than $\eps$ then $\tilde h$ has error of less than $1/4$. Learning all Boolean functions to error of at most $1/4$ requires $2^{\Omega(k)} = 2^{\Omega(2^r)} = 2^{\Omega(1/\eps^2)}$ random examples (or value queries) and therefore we obtain our claim.
\end{proof}

We now observe that there exists some constant $c$ such that $\aedeg(f) \leq c/\epsilon^2$. Otherwise, for any constant $c_0$, using Theorem \ref{thm:low-degree-regression} we could obtain an algorithm that learns self-bounding functions using $2^{c_1/\eps^2}$ random examples contradicting Theorem \ref{thm:low-degree-regression}.

\begin{theorem}
For every $\eps>0$ there exists $k =O(1/\eps^2)$ and a self-bounding function $f:\zo^{k} \rightarrow [0,1]$ such that $\aedeg(f)=\Omega(1/\epsilon^2)$.
\end{theorem}

\fi

\paragraph*{Acknowledgements}
We would like to thank Pravesh Kothari for useful discussions and his help with the proof of Lemma \ref{lem:submod-conc-lower}.

\bibliographystyle{alpha}
\bibliography{l2degreerefs}

\newcommand{\etalchar}[1]{$^{#1}$}
\begin{thebibliography}{DSFT{\etalchar{+}}15}

\bibitem[AA14]{AaronsonA:14}
Scott Aaronson and Andris Ambainis.
\newblock The need for structure in quantum speedups.
\newblock {\em Theory of Computing}, 10:133--166, 2014.

\bibitem[AM06]{AmanoM06}
Kazuyuki Amano and Akira Maruoka.
\newblock On learning monotone boolean functions under the uniform
  distribution.
\newblock {\em Theor. Comput. Sci.}, 350(1):3--12, 2006.

\bibitem[BB14]{BackursB:14}
Arturs Backurs and Mohammad Bavarian.
\newblock On the sum of {L1} influences.
\newblock In {\em {CCC}}, pages 132--143, 2014.

\bibitem[BBL98]{BlumBL:98}
A.~Blum, C.~Burch, and J.~Langford.
\newblock On learning monotone boolean functions.
\newblock In {\em FOCS}, pages 408--415, 1998.

\bibitem[BCIW12]{BalcanCIW:12}
M.F. Balcan, F.~Constantin, S.~Iwata, and L.~Wang.
\newblock Learning valuation functions.
\newblock {\em COLT}, 23:4.1--4.24, 2012.

\bibitem[BDF{\etalchar{+}}12]{BadanidiyuruDFKNR:12}
A.~Badanidiyuru, S.~Dobzinski, Hu~Fu, R.~Kleinberg, N.~Nisan, and
  T.~Roughgarden.
\newblock Sketching valuation functions.
\newblock In {\em SODA}, pages 1025--1035, 2012.

\bibitem[BH12]{BalcanHarvey:12full}
M.F. Balcan and N.~Harvey.
\newblock Submodular functions: Learnability, structure, and optimization.
\newblock {\em CoRR}, abs/1008.2159, 2012.
\newblock Earlier version in STOC 2011.

\bibitem[BLB03]{BoucheronLB03}
St{\'{e}}phane Boucheron, G{\'{a}}bor Lugosi, and Olivier Bousquet.
\newblock Concentration inequalities.
\newblock In {\em Advanced Lectures on Machine Learning, {ML} Summer Schools
  2003, Revised Lectures}, pages 208--240, 2003.

\bibitem[BLM00]{BoucheronLM:00}
S.~Boucheron, G.~Lugosi, and P.~Massart.
\newblock A sharp concentration inequality with applications.
\newblock {\em Random Struct. Algorithms}, 16(3):277--292, 2000.

\bibitem[BLN06]{LLN06}
D.~J.~Lehmann B.~Lehmann and N.~Nisan.
\newblock Combinatorial auctions with decreasing marginal utilities.
\newblock {\em Games and Economic Behavior}, 55:1884--1899, 2006.

\bibitem[BM02]{BartlettMendelson:02}
P.~Bartlett and S.~Mendelson.
\newblock {Rademacher and Gaussian} complexities: Risk bounds and structural
  results.
\newblock {\em JMLR}, 3:463--482, 2002.

\bibitem[BOL85]{Ben-OrLinial:85}
M.~Ben-Or and N.~Linial.
\newblock Collective coin flipping, robust voting schemes and minima of banzhaf
  values.
\newblock In {\em FOCS}, pages 408--416, 1985.

\bibitem[BOSY13]{BlaisOSY:13manu}
E.~Blais, K.~Onak, R.~Servedio, and G.~Yaroslavtsev.
\newblock Concise representations of discrete submodular functions, 2013.
\newblock Personal communication.

\bibitem[BRY14]{BermanRY14}
Piotr Berman, Sofya Raskhodnikova, and Grigory Yaroslavtsev.
\newblock L\({}_{\mbox{p}}\)-testing.
\newblock In {\em {STOC}}, pages 164--173, 2014.

\bibitem[BT96]{BshoutyTamon:96}
N.~Bshouty and C.~Tamon.
\newblock {On the Fourier spectrum of monotone functions}.
\newblock {\em Journal of the ACM}, 43(4):747--770, 1996.

\bibitem[CKKL12]{CheraghchiKKL:12}
M.~Cheraghchi, A.~Klivans, P.~Kothari, and H.~Lee.
\newblock Submodular functions are noise stable.
\newblock In {\em SODA}, pages 1586--1592, 2012.

\bibitem[DFKO06]{DinurFKO:06}
I.~Dinur, E.~Friedgut, G.~Kindler, and R.~O'Donnell.
\newblock {On the Fourier tails of bounded functions over the discrete cube}.
\newblock In {\em STOC}, pages 437--446, 2006.

\bibitem[DLM{\etalchar{+}}08]{Dachman-SoledLMSWW08}
Dana Dachman{-}Soled, Homin~K. Lee, Tal Malkin, Rocco~A. Servedio, Andrew Wan,
  and Hoeteck Wee.
\newblock Optimal cryptographic hardness of learning monotone functions.
\newblock In {\em {ICALP} Track {A:} Algorithms, Automata, Complexity, and
  Games}, pages 36--47, 2008.

\bibitem[DS06]{DS06}
S.~Dobzinski and M.~Schapira.
\newblock An improved approximation algorithm for combinatorial auctions with
  submodular bidders.
\newblock In {\em SODA}, pages 1064--1073, 2006.

\bibitem[DSFT{\etalchar{+}}15]{DachmanFTWW:15}
Dana Dachman-Soled, Vitaly Feldman, Li-Yang Tan, Andrew Wan, and Karl Wimmer.
\newblock Approximate resilience, monotonicity, and the complexity of agnostic
  learning.
\newblock In {\em SODA}, 2015.

\bibitem[Edm70]{E70}
Jack Edmonds.
\newblock Matroids, submodular functions and certain polyhedra.
\newblock {\em Combinatorial Structures and Their Applications}, pages 69--87,
  1970.

\bibitem[Fei06]{Feige:06}
Uriel Feige.
\newblock On maximizing welfare when utility functions are subadditive.
\newblock In {\em ACM STOC}, pages 41--50, 2006.

\bibitem[FFI01]{FFI00}
L.~Fleischer, S.~Fujishige, and S.~Iwata.
\newblock A combinatorial, strongly polynomial-time algorithm for minimizing
  submodular functions.
\newblock {\em JACM}, 48(4):761--777, 2001.

\bibitem[FK14]{FeldmanK14}
V.~Feldman and P.~Kothari.
\newblock Learning coverage functions and private release of marginals.
\newblock In {\em {COLT}}, pages 679--702, 2014.

\bibitem[FKV13]{FeldmanKV:13}
V.~Feldman, P.~Kothari, and J.~Vondr\'ak.
\newblock Representation, approximation and learning of submodular functions
  using low-rank decision trees.
\newblock In {\em COLT}, 2013.

\bibitem[FKV14]{FeldmanKV14}
V.~Feldman, P.~Kothari, and J.~Vondr{\'{a}}k.
\newblock Nearly tight bounds on
  {\textdollar}{\textbackslash}ell{\_}1{\textdollar} approximation of
  self-bounding functions.
\newblock {\em CoRR}, abs/1404.4702, 2014.

\bibitem[FMV07]{FeigeMV:07}
U.~Feige, V.~Mirrokni, and J.~Vondr{\'a}k.
\newblock Maximizing non-monotone submodular functions.
\newblock In {\em FOCS}, pages 461--471, 2007.

\bibitem[Fra97]{F97}
Andr\'as Frank.
\newblock Matroids and submodular functions.
\newblock {\em Annotated Biblographies in Combinatorial Optimization}, pages
  65--80, 1997.

\bibitem[Fri98]{Friedgut:98}
E.~Friedgut.
\newblock Boolean functions with low average sensitivity depend on few
  coordinates.
\newblock {\em Combinatorica}, 18(1):27--35, 1998.

\bibitem[FV13]{FeldmanVondrak:13arxiv}
V.~Feldman and J.~Vondr\'ak.
\newblock Optimal bounds on approximation of submodular and {XOS} functions by
  juntas.
\newblock {\em CoRR}, abs/1307.3301, 2013.
\newblock Earlier version in FOCS 2013.

\bibitem[GGL{\etalchar{+}}00]{GoldreichGLRS00}
Oded Goldreich, Shafi Goldwasser, Eric Lehman, Dana Ron, and Alex
  Samorodnitsky.
\newblock Testing monotonicity.
\newblock {\em Combinatorica}, 20(3):301--337, 2000.

\bibitem[GHIM09]{GHIM09}
M.~Goemans, N.~Harvey, S.~Iwata, and V.~Mirrokni.
\newblock Approximating submodular functions everywhere.
\newblock In {\em SODA}, pages 535--544, 2009.

\bibitem[GHRU11]{GuptaHRU:11}
A.~Gupta, M.~Hardt, A.~Roth, and J.~Ullman.
\newblock Privately releasing conjunctions and the statistical query barrier.
\newblock In {\em STOC}, pages 803--812, 2011.

\bibitem[GKS05]{GKS05}
C.~Guestrin, A.~Krause, and A.~Singh.
\newblock Near-optimal sensor placements in gaussian processes.
\newblock In {\em ICML}, pages 265--272, 2005.

\bibitem[GV06]{GoemansVondrak:06}
M.~Goemans and J.~Vondr{\'a}k.
\newblock Covering minimum spanning trees of random subgraphs.
\newblock {\em Random Struct. Algorithms}, 29(3):257--276, 2006.

\bibitem[KGGK06]{KGGK06}
A.~Krause, C.~Guestrin, A.~Gupta, and J.~Kleinberg.
\newblock Near-optimal sensor placements: maximizing information while
  minimizing communication cost.
\newblock In {\em IPSN}, pages 2--10, 2006.

\bibitem[KK07]{KahnKalai07}
J.~Kahn and G.~Kalai.
\newblock Thresholds and expectation thresholds.
\newblock {\em Combinatorics, Probability and Computing}, 16(3):492--502, 2007.

\bibitem[KKL88]{KahnKL:88}
J.~Kahn, G.~Kalai, and N.~Linial.
\newblock {The influence of variables on Boolean functions}.
\newblock In {\em FOCS}, pages 68--80, 1988.

\bibitem[KKM13]{KaneKM13}
Daniel~M. Kane, Adam Klivans, and Raghu Meka.
\newblock Learning halfspaces under log-concave densities: Polynomial
  approximations and moment matching.
\newblock In {\em {COLT}}, pages 522--545, 2013.

\bibitem[KKMS08]{KalaiKMS:08}
A.~Kalai, A.~Klivans, Y.~Mansour, and R.~Servedio.
\newblock Agnostically learning halfspaces.
\newblock {\em SIAM Journal on Computing}, 37(6):1777--1805, 2008.

\bibitem[KLV94]{KearnsLV:94}
M.~Kearns, M.~Li, and L.~Valiant.
\newblock Learning boolean formulas.
\newblock {\em Journal of the ACM}, 41(6):1298--1328, 1994.

\bibitem[KP00]{KoltchinskiiPanchenko:00}
Vladimir Koltchinskii and Dmitriy Panchenko.
\newblock Rademacher processes and bounding the risk of function learning.
\newblock In {\em High Dimensional Probability II}, volume~47 of {\em Progress
  in Probability}, pages 443--457. Birkhauser Boston, 2000.

\bibitem[KS08]{KlivansS08}
Adam~R. Klivans and Rocco~A. Servedio.
\newblock Learning intersections of halfspaces with a margin.
\newblock {\em J. Comput. Syst. Sci.}, 74(1):35--48, 2008.

\bibitem[KSG08]{KSG08}
A.~Krause, A.~Singh, and C.~Guestrin.
\newblock Near-optimal sensor placements in gaussian processes: Theory,
  efficient algorithms and empirical studies.
\newblock {\em JMLR}, 9:235--284, 2008.

\bibitem[KSS94]{KearnsSS:94}
M.~Kearns, R.~Schapire, and L.~Sellie.
\newblock Toward efficient agnostic learning.
\newblock {\em Machine Learning}, 17(2-3):115--141, 1994.

\bibitem[KST08]{KakadeST:08}
S.~Kakade, K.~Sridharan, and A.~Tewari.
\newblock On the complexity of linear prediction: Risk bounds, margin bounds,
  and regularization.
\newblock In {\em NIPS}, pages 793--800, 2008.

\bibitem[LMN93]{LinialMN:93}
N.~Linial, Y.~Mansour, and N.~Nisan.
\newblock {Constant depth circuits, Fourier transform and learnability}.
\newblock {\em Journal of the ACM}, 40(3):607--620, 1993.

\bibitem[Lov83]{L83}
L\'aszl\'o Lov\'asz.
\newblock Submodular functions and convexity.
\newblock {\em Mathematical Programmming: The State of the Art}, pages
  235--257, 1983.

\bibitem[MO02]{mossel2002noise}
Elchanan Mossel and Ryan O'Donnell.
\newblock On the noise sensitivity of monotone functions.
\newblock In {\em Mathematics and Computer Science II}, pages 481--495.
  Springer, 2002.

\bibitem[MR06]{MR06}
C.~McDiarmid and B.~Reed.
\newblock Concentration for self-bounding functions and an inequality of
  talagrand.
\newblock {\em Random structures and algorithms}, 29:549--557, 2006.

\bibitem[O'D03]{Odonnell:03thesis}
R.~O'Donnell.
\newblock {\em Computational Applications of Noise Sensitivity}.
\newblock PhD thesis, 2003.

\bibitem[O'D14]{ODonnell14:book}
Ryan O'Donnell.
\newblock {\em Analysis of Boolean Functions}.
\newblock Cambridge University Press, 2014.

\bibitem[OW13]{OWimmer:09}
Ryan O'Donnell and Karl Wimmer.
\newblock {KKL}, {K}ruskal-{K}atona, and monotone nets.
\newblock {\em SIAM J. Comput.}, 42(6):2375--2399, 2013.

\bibitem[Que95]{Q95}
Maurice Queyranne.
\newblock A combinatorial algorithm for minimizing symmetric submodular
  functions.
\newblock In {\em SODA}, pages 98--101, 1995.

\bibitem[RY13]{RaskhodnikovaYaroslavtsev:13}
S.~Raskhodnikova and G.~Yaroslavtsev.
\newblock Learning pseudo-boolean {k-DNF} and submodular functions.
\newblock In {\em SODA}, 2013.

\bibitem[SV11]{SV11}
C.~Seshadhri and J.~Vondr\'ak.
\newblock Is submodularity testable?
\newblock In {\em Innovations in computer science}, pages 195--210, 2011.

\bibitem[Tal94]{Talagrand:94}
M.~Talagrand.
\newblock On {R}usso's approximate zero-one law.
\newblock {\em The Annals of Probability}, pages 1576--1587, 1994.

\bibitem[Tal96]{Talagrand:96}
M.~Talagrand.
\newblock How much are increasing sets positively correlated?
\newblock {\em Combinatorica}, 16(2):243--258, 1996.

\bibitem[Val84]{Valiant:84}
L.~G. Valiant.
\newblock A theory of the learnable.
\newblock {\em Communications of the ACM}, 27(11):1134--1142, 1984.

\bibitem[Von08]{Vondrak08}
J.~Vondr{\'a}k.
\newblock Optimal approximation for the submodular welfare problem in the value
  oracle model.
\newblock In {\em STOC}, pages 67--74, 2008.

\end{thebibliography}
\iffull
\appendix
\section{Rademacher complexity, XOS and self-bounding functions}
\label{sec:rademacher}
The Rademacher complexity of a class of functions $\F$ is one of the most popular and powerful tools in statistical learning theory for proving uniform convergence bounds on the generalization error \cite{KoltchinskiiPanchenko:00,BartlettMendelson:02}. Specifically, for a possibly unknown distribution $\cP$ over some domain $X$ we would like to upper bound the value of $n$ for which
$$\pr_{x^1,\ldots,x^n \sim \cP}\left[\sup_{f \in \F} \left| \E_{x\sim \cP}[f(x)] - \frac{1}{n}\sum_{i\in [n]} f(x^i)\right| \geq \eps \right] \leq \delta . $$
Such bounds imply learnability via empirical loss minimization and can be obtained by considering the Rademacher complexity of $\F$ relative to $\cP$ which is defined as follows: for a (multi-)set $S$ of $n$ points from $X$ let the empirical Rademacher complexity be defined as
$$ \cR(\cF \circ S) = \frac{1}{n} \E_{\sigma \sim \on^n} \left[ \sup_{f \in \F} \sum_{i\in [n]} \sigma_i f(x^i) \right] ,$$
where $\sigma$ is distributed uniformly over $\on^n$, or equivalently each $\sigma_i$ is an independent Rademacher variable. More generally, Rademacher complexity of any bounded set of vectors $V \subseteq \R^n$ is defined as
$$ \cR(V) = \frac{1}{n} \E_{\sigma \sim \on^n} \left[ \sup_{v \in V} \sum_{i\in [n]} \sigma_i v_i \right] .$$
The Rademacher complexity of $\cF$ over $\cP$ for sample size $n$ is then $\cR_n(\cF,\cP) = \E_{S \sim \cP^n} [\cR(\cF \circ S) ]$. To study the concentration properties of empirical Rademacher complexity it is viewed as a function over subsets of $[n]$ defined as $$\cR(\cF \circ S,A) = \frac{1}{n}\E_{\sigma \sim \on^{A}} \left[ \sup_{f \in \F} \sum_{i\in A} \sigma_i f(x^i) \right],$$
in other words it measures the Rademacher complexity of $S$ restricted to points with indices in $A$. This function is known to be self-bounding ---  an essential property for the applications of Rademacher complexity that rely on strong concentration of measure results (\eg \cite{BoucheronLB03}). Here we show that Rademacher complexity of any set of vectors $V$ is in fact an XOS function. For completeness, in Section \ref{sec:separate} we show that this is a strictly smaller class than that of monotone self-bounding functions. 
\subsection{Equivalence of XOS and Rademacher complexity functions}
For convenience we remove the normalizing factor $\frac{1}{n}$ in the definition of the Rademacher complexity since it does not affect the membership of a function in XOS.
\begin{theorem}
\label{thm:rademacher-is-XOS}
Let $V$ be a bounded set of vectors from $\R^n$. Then function $\phi:2^{[n]} \rightarrow \R$ defined as
$$\cR(V,A) = \frac{1}{n}\E_{\sigma \sim \on^{A}} \left[ \sup_{v \in V} \sum_{i\in A} \sigma_i v_i \right]$$ is XOS.
\end{theorem}
\begin{proof}
For convenience we prove that $\phi(A) = n\cdot \cR(V,A)$ is XOS which naturally implies that $\cR(V,A)$ is XOS.
We first observe that $\phi(\emptyset) = 0$. Next we show that $\phi$ is monotone.
For simplicity we assume that $V$ is compact (the extension to general sets is straightforward). For $A$ and a vector $\sigma \in \on^A$ let $v^{A,\sigma} \in V$ be a vector such that $\sum_{i\in A} \sigma_i v^{A,\sigma}_i = \sup_{v \in V} \sum_{i\in A} \sigma_i v_i$.
For subsets $A \subset A' \subset [n]$ and $\sigma' \in \on^{A'}$ we denote by $\sigma'_A$ the vector containing the bits of $\sigma'$ with indices in $A$. Then
\alequn{\phi(A') & = \E_{\sigma' \sim \on^{A'}} \left[ \sup_{v \in V} \sum_{i\in A'} \sigma'_i v_i \right] \geq \E_{\sigma' \sim \on^{A'}} \left[ \sum_{i \in A'} \sigma'_i v^{A,\sigma'_A}_i \right]
\\ & = \E_{\sigma' \sim \on^{A'}} \left[ \sum_{i \in A} \sigma'_i v^{A,\sigma'_A}_i \right] + \E_{\sigma' \sim \on^{A'}} \left[\sum_{i \in A'\setminus A} \sigma'_i v^{A,\sigma'_A}_i \right]
\\ & = \E_{\sigma \sim \on^{A}} \left[ \sum_{i \in A} \sigma_i v^{A,\sigma}_i \right]
\\ & = \E_{\sigma \sim \on^{A}} \left[ \sup_{v \in V} \sum_{i\in A} \sigma_i v_i \right] = \phi(A),
}
where we used the fact that $\sigma'_i$ for $i\in A'\setminus A$ is independent of $v^{A,\sigma'_A}_i$.

The function $\phi$ has non-negative range and $\phi(\emptyset) = 0$. Therefore it is sufficient to prove that  $\phi$ is fractionally subadditive. That is we need to prove that $\phi(A) \leq \sum_{j\in [m]} \beta_j \phi(B_j)$ whenever $\beta_j \geq 0$ and $\sum_{j:i \in B_j} \beta_j \geq 1 \ \forall i \in A$. Monotonicity of $\phi$ implies that it is sufficient to prove this condition for exact fractional covers: that is $\sum_{j:i \in B_j} \beta_j = 1 \ \forall i \in A$. This condition implies that for every vector $w \in \R^n$,
$$\sum_{j\in [m]} \sum_{i\in B_j} \beta_j w_i = \sum_{i\in A} w_i .$$ Using this equality we can conclude:
\alequn{\phi(A) & = \E_{\sigma \sim \on^{A}} \left[ \sum_{i \in A} \sigma_i v^{A,\sigma}_i \right]\\
& = \E_{\sigma \sim \on^{A}} \left[ \sum_{j \in [m]}\sum_{i \in B_j} \beta_j \cdot  \sigma_i v^{A,\sigma}_i \right] \\
& = \sum_{j \in [m]} \left(\beta_j \cdot  \E_{\sigma \sim \on^{A}} \left[ \sum_{i \in B_j} \sigma_i v^{A,\sigma}_i \right] \right)\\
& \leq  \sum_{j \in [m]} \left(\beta_j \cdot \E_{\sigma \sim \on^{A}} \left[ \sum_{i \in B_j} \sigma_i v^{B_j,\sigma_{B_j}}_i \right] \right) \\ & = \sum_{j \in [m]} \left(\beta_j \cdot \E_{\sigma \sim \on^{B_j}} \left[ \sum_{i \in B_j} \sigma_i v^{B_j,\sigma}_i \right] \right) = \sum_{j \in [m]} \beta_j \cdot \phi(B_j).
}
\end{proof}

We remark that this proof also applies to Gaussian complexity of a set of vectors $V$, another measure of complexity studied in convex geometry and statistical learning theory. In this measure in place of a Rademacher variable, a 0-mean Gaussian with variance 1 is used (the only fact about $\sigma_i$'s that we used is that it is $0$-mean and independent of all other variables).

It turns out that the converse of Theorem \ref{thm:rademacher-is-XOS} is also true. Any XOS function can be represented as Rademacher complexity of some set of vectors.
\begin{theorem}
Let $f:2^{[n]} \rightarrow \R$ be an XOS function. Then there exists a set $V$ such that for every set $A \subseteq [n]$,
$$f(A) = \frac{1}{n} \E_{\sigma \sim \on^{A}} \left[ \max_{v \in V} \sum_{i\in A} \sigma_i v_i \right].$$
\end{theorem}
\begin{proof}
By definition of XOS, there exists a set of clauses $C$ such that
$f(A) = \max_{c \in C} \sum_{i \in A}^{n} w_{ci}$, for some non-negative weights $w_{ci}$.
Let $$V = \{\frac{1}{n} (w_{c1} \sigma_1,w_{c2} \sigma_2,\ldots, w_{cn} \sigma_n ) \cond c\in C,\ \sigma \in \on^n\}.$$
Then for every $A$ and $\sigma$,
$$\max_{v \in V} \sum_{i\in A} \sigma_i v_i= \max_{v \in V} \sum_{i\in A} |v_i| = \max_{c \in C} \sum_{i\in A} w_{ci} =  n \cdot f(A).$$ This implies that
$$\frac{1}{n} \E_{\sigma \sim \on^{A}} \left[ \max_{v \in V} \sum_{i\in A} \sigma_i v_i \right] = f(A).$$
\end{proof}

\subsection{Separation of XOS and monotone self-bounding functions}
\label{sec:separate}
Here we show a simple monotone self-bounding function which is not XOS. We remark that formally, such a separation is trivial since XOS functions must satisfy $f({\bf 0}) = 0$,  unlike monotone self-bounding functions. Here we present a more interesting example, a function $f:\zo^3 \rightarrow \RR_+$ such that $f$ is 1-Lipschitz monotone self-bounding, $f({\bf 0}) = 0$ and $f$ is not XOS. The function is defined as follows (in set notation):
\begin{compactitem}
\item $f(\emptyset) = 0$
\item $f(\{1\}) = 1/5, f(\{2\}) = 2/5, f(\{3\}) = 3/5$
\item $f(\{1,2\}) = 3/5, f(\{1,3\}) = 4/5, f(\{2,3\}) = 3/5$
\item $f(\{1,2,3\}) = 1$
\end{compactitem}
The reader can verify that this function is monotone self-bounding but not XOS (in fact not even subadditive, since $f(\{1,2\}) + f(\{3\}) > f(\{1,2,3\})$). 
\fi

\end{document}